\documentclass[a4paper]{article}

\usepackage{appendix}

\usepackage{amsmath, amssymb, amsfonts}

\usepackage{bbm}
\usepackage{graphicx}
\usepackage{caption}
\usepackage{subcaption}
\graphicspath{{./figures/}}
\usepackage{authblk}

\usepackage[normalem]{ulem}
\usepackage[usenames,svgnames]{xcolor}

\usepackage{hyperref}
\hypersetup{
     colorlinks=true,           % false: boxed links; true: colored links
     linkcolor=Salmon,          % color of internal links
     citecolor=blue,            % color of links to bibliography
     filecolor=blue,            % color of file links
     urlcolor=cyan,             % color of external links
 }

\newcommand{\Eq}[1]{Eq.~(\ref{#1})}
\newcommand{\Def}[1]{Definition~\ref{#1}}
\newcommand{\Lem}[1]{Lemma~\ref{#1}}
\newcommand{\cc}[1]{~\cite{#1}}
\newcommand{\cRef}[1]{Ref.~\cite{#1}}

\DeclareMathOperator{\gap}{gap}

\DeclareMathOperator{\supp}{Supp}

\DeclareMathOperator{\diag}{diag}
\DeclareMathOperator{\Prob}{Prob}

\DeclareMathOperator*{\Tr}{Tr}

\renewcommand{\Re}{\textrm{Re}}

\newcommand{\poly}{\mathrm{poly}} 
 
\newcommand{\EqDef}{\stackrel{\mathrm{def}}{=}}
\newcommand\tab[1][1cm]{\hspace*{#1}}

\newcommand{\Id}{\mathbbm{1}}

\newcommand{\BBR}{\mathbbm{R}}
\newcommand{\BBH}{\mathbbm{H}}

\newcommand{\mcH}{\mathcal{H}}
\newcommand{\mch}{\mathcal{H}}

\newcommand{\mcB}{\mathcal{B}}
\newcommand{\mcI}{\mathcal{I}}

\newcommand{\mcF}{\mathcal{F}}
\newcommand{\mcL}{\mathcal{L}}
\newcommand{\mcT}{\mathcal{T}}
\newcommand{\mcG}{\Delta}
\newcommand{\mcP}{\mathcal{P}}
\newcommand{\gs}{\Omega}
\newcommand{\bi}{b}

\newcommand{\glg}{\Gamma^{-1/2}\circ \mcL \circ \Gamma^{1/2}}

\newcommand{\gkl}{\gamma_{k,\bi}}

\newcommand{\LLkl}{\mcL_{k,\bi}}
\newcommand{\hkl}{\mcH_{k,\bi}}
\newcommand{\gin}{\gamma^{\mathrm{in}}}
\newcommand{\gout}{\gamma^{\mathrm{out}}}

\newcommand{\ux}{{\underline{x}}}

\newcommand{\norm}[1]{{\| #1 \|}}

\newcommand{\ket}[1]{{ |{#1} \rangle }}  
\newcommand{\bra}[1]{{\langle {#1} | }}
\newcommand{\braket}[1]{{ \langle {#1} \rangle }}  
\newcommand{\ketbra}[2]{{ |{#1} \rangle\langle {#2} | }}

\newcommand{\dket}[1]{|#1\rangle\!\rangle}
\newcommand{\dbra}[1]{\langle\!\langle #1|}
\newcommand{\dbraket}[1]{\langle\!\langle #1 \rangle\!\rangle}
\newcommand{\dketbra}[2]{|#1\rangle\!\rangle \langle\!\langle#2|}

\newcommand{\bigO}[1]{\ensuremath{\operatorname{O}\bigl(#1\bigr)}}

\newcommand{\bOmega}[1]{\ensuremath{\operatorname{\Omega}\bigl(#1\bigr)}}

\newcommand{\bTheta}[1]{\ensuremath{\operatorname{\Theta}\bigl(#1\bigr)}}

\newtheorem{theorem}{Theorem}[section]

\newtheorem{definition}[theorem]{Definition}  
  
\newtheorem{claim}[theorem]{Claim}  
\newtheorem{lemma}[theorem]{Lemma}

\newtheorem{corol}[theorem]{Corollary}

\newcommand{\qedsymb}{\hfill{\rule{2mm}{2mm}}}  
\newenvironment{proof}[1][]{\begin{trivlist}  
\item[\hspace{\labelsep}{\bf\noindent Proof#1:\/}]}{\qedsymb\end{trivlist}}

\title{\textbf{Area law for steady states of detailed-balance
local Lindbladians}}

\author[1]{Raz Firanko}
\author[2]{Moshe Goldstein}
\author[1]{Itai Arad}
\affil[1]{Physics Department, Technion, Haifa 3200003, Israel}
\affil[2]{Raymond and Beverly Sackler School of Physics and
  Astronomy, Tel-Aviv University, Tel Aviv 6997801, Israel}

\begin{document}
\maketitle

%%%%%%%%%%%%%%%%%%%%%%%%%%%% abstract %%%%%%%%%%%%%%%%%%%%%%%%%
\begin{abstract}
  We study steady-states of quantum Markovian processes whose
  evolution is described by local Lindbladians.  We assume that the
  Lindbladian is gapped and satisfies quantum detailed balance with
  respect to a unique full-rank steady state $\sigma$. We show that 
  under mild assumptions on the Lindbladian terms, which can be
  checked efficiently, the Lindbladian can be mapped to a local
  Hamiltonian on a doubled Hilbert space that has the same spectrum,
  and a ground state that is the vectorization of $\sigma^{1/2}$.
  Consequently, we can use Hamiltonian complexity tools to study the
  steady states of such open systems. In particular, we show an
  area-law in the mutual information for the steady state of such 1D
  systems, together with a tensor-network representation that can be
  found efficiently. 
\end{abstract}

\tableofcontents

%%%%%%%%%%%%%%%%%%%%%%%%% intro %%%%%%%%%%%%%%%%%%%%%%%%%
\section{Introduction} 
\label{sec:Introduction} 

Understanding the structure and physical properties of open
many-body quantum systems is a major problem in condensed matter
physics. Over the years, these systems have been studied using a
plethora of methods from statistical physics, many-body quantum
theory and functional analysis. More recently, with the advent of
quantum computation and quantum information, these systems have also
been studied using techniques coming from quantum information.
Coming from these research paradigms, one would typically want to
understand the computational complexity of these systems, the type
of entanglement and correlations that they can create, and whether
or not they can be represented efficiently on a classical computer.

This line of research has been well-established in the context of
\emph{closed systems}. In particular, there are many results
characterizing the complexity of ground states of local Hamiltonians
defined on a lattice. For example, it has been shown that for
several families of Hamiltonians it is QMA-hard to approximate the
ground state energy\cc{ref:Kitaev2002-QCbook, ref:Kempe2006-QMA},
implying that the ground state of such systems does not posses an efficient 
classical description unless QMA=NP. On the other hand, it is generally
believed that gapped local Hamiltonians on a $D$-dimensional lattice
satisfy an area-law of entanglement
entropy\cc{ref:Eisert2010-AL-review} and can be well approximated by
efficient tensor network states. This has been rigorously proven in
1D\cc{ref:Hastings2007-1DAL}, with some partial results in higher
dimension\cc{ref:Masanes2009-AL, ref:deBeaudrap2010-ALFFspins,
ref:Spyridon2012-adiabaticAL, ref:Jaeyoon2014-AL,
ref:Brandao2015-sp-heatAL, ref:Anshu2022-subvol,
ref:Anshu2021-2DAL}.

A natural question to ask is whether, and to what extent, tools and
techniques that are used to analyze the complexity of ground states
of closed systems can be used to study the complexity of steady
states of open systems. Specifically, in this paper we consider open
systems that are described by Markovian dynamics that is generated
by a local Lindbladian superoperator $\mcL=\sum_i \mcL_i$ (see
\autoref{sec:background} for a formal definition). The steady state
of the system is then given by a density operator $\sigma$ for which

\begin{align*}
  \mcL(\sigma)=\sum_i \mcL_i(\sigma)=0 . 
\end{align*}
This is a homogeneous linear equation, which has a striking
similarity to the corresponding problem in closed system 
\begin{align*}
  (H-\epsilon_0\Id)\ket{\gs}
    =\left(\sum_i h_i-\epsilon_0\Id\right)\ket{\gs}=0 ,
\end{align*}
where $H=\sum_i h_i$ is a local Hamiltonian with ground state
$\ket{\gs}$ and ground energy $\epsilon_0$. Moreover, as described
in \autoref{sec:background}, the spectrum of
$\mcL$ is in the $\Re\le 0$ part of the complex plane, and we can
define the spectral gap of the Linbladian to be the
largest non-zero real part of the spectrum. It is therefore tempting
to try and use local Hamiltonian techniques to characterize
$\sigma$. However, a quick inspection reveals the main obstacle for
such a simple plan to work: Whereas the Hamiltonian is a
self-adjoint (hermitian) operator, which therefore has a real
spectrum with a set of orthonormal eigenstates, the same is not true
for the Lindbladian; it is generally not a self-adjoint operator,
and consequently its spectrum might be complex with non-orthogonal
eigenoperators. To make $\mcL$ self-adjoint, one might consider
$\mcL^*\mcL$, but this come at the price of losing locality.
Alternatively, we might consider $\frac{1}{2}(\mcL + \mcL^*)$, but
it is not clear how the eigenstates of this operator are related to
the eignstates of $\mcL$.

A more sophisticated way of obtaining self-adjointness is by
assuming that the Linbladian satisfies quantum
detailed-balance\cc{ref:Agarwal1973-QDB, ALICKI1976249}. For such
Linbladians there is a way of defining an inner product with respect
to which $\mcL$ is self-adjoint\cc{ALICKI1976249,Fangola, chi, Maas}
(see \autoref{sec:QDB}). This inner product, however, can be highly
non-local (with respect to the underlying tensor-product structure)
and might deform the natural geometry of the Hilbert space. It is
therefore not clear how to use it to bound correlations and other
measures of locality in the steady state.

Nevertheless, in this paper we show that the quantum
detailed-balance condition provides a surprisingly simple map that
takes a local Lindbladian to a self-adjoint \emph{local}
super-operartor (i.e, a super Hamiltonian), which is then mapped to
a local Hamiltonian using vectorization. All this is done while
maintaining a direct relation between the steady state of the former
and the ground state of the latter. Consequently, many of the bounds
and properties that were proved for ground states, easily transform
to the steady states of detailed-balance Lindbladians. In
particular, we show that under mild conditions, which can be
efficiently checked, steady states of gapped 1D Lindbladian with
detailed balance satisfy an area-law for the mutual information.
Moreover, just as in the 1D area-law case for local Hamiltonians,
these steady states are well-approximated by an efficient tensor
network.  

We conclude this section by noting that our mapping is not new; it
has already been used before, e.g., in Refs.\cc{alicki2009, chi,DOC,
kastoryano2016quantum,Szegedy}. Our main contribution is showing
that this mapping results in a \emph{local} Hamiltonian, rather than
a general Hermitian operator, which enables us to take advantage of
the local Hamiltonian machinery (see \autoref{sec:mapping}).

%===============================================================
\subsection{Comparison with previous works}
\label{sec:previous-results}

Several works studied the entanglement structure of steady states of
open systems, see for example Refs.~\cc{DOC, Brandao,
simulatability, Swingle2016entanglement}. Here we will concentrate
on Refs.~\cc{Brandao,DOC}, which study problems that are very close
to ours.

In \cRef{DOC} the authors considered steady states of local
Lindbladians on a $D$-dimensional lattice with unique, full-rank
steady state. Under the assumption of gapped, detailed-balanced
Linbladian, they have shown exponential decay of correlations in the
steady state. To prove an area law, the authors additionally assumed
a very strong form of fast convergence to the steady state, known as
system-size independent log-Sobolev constant (see \cRef{logsobolev}
for a definition). Under this assumption, and using the
Lieb-Robinson bounds for open systems\cc{ref:Poulin2010-openLR,
ref:Nachtergaele2011-openLR, ref:Barthel2012-openLR}, they showed
clustering of correlations in terms of the mutual information, and
as a consequence, an area law for the mutual information. Specifically, for every region
$A$ in the lattice, the mutual information is bounded by
\begin{align}
\label{eq:Sobolev-AL}
  I(A:A^c) \leq c\log\log(\norm{\sigma^{-1}})|\partial A|.
\end{align}
Notice, however, that for full rank steady states, $\norm{\sigma^{-1}}$ grows
at least exponentially with system size, and might even grow doubly
exponential with system size, in which case \eqref{eq:Sobolev-AL} is
no longer an area-law.

In \cRef{Brandao} the authors proved an area-law for the steady
state of what they refer to as a uniform family of Linbladians.
Essentially, this is a local Linbladian defined on an infinite
lattice, together with a set of boundary conditions that allow one
to restrict the dynamics to finite regions $\{\Lambda\}$ in the
infinite lattice. In this setup, the authors assumed another strong
form of convergence, which is called \emph{rapid mixing}. Roughly,
it assumes that the restriction of the system to any region
$\Lambda$ on the lattice has a unique fixed point $\sigma_\Lambda$,
to which it converges exponentially fast \emph{from any initial
state},
\begin{align}
  \norm{\rho(t)-\sigma_\Lambda}_1
    \le c|\Lambda|^\delta e^{-\gamma t},
\end{align}
where $c, \delta, \gamma$ are some region-independent constants. To
prove an area-law, they have also assumed that the Linbladians are
frustration-free, or that they posses pure steady-states. Under
these assumptions, using Lieb-Robinson bounds, they have managed to
show the following area-law with logarithmic corrections
\begin{align}
\label{eq:LS-AL-bound}
  I(A:A^c) \le  c|\partial A|\log(|A|) .
\end{align}

It is interesting to contrast these two results with several
area-law results for ground states of gapped local
Hamiltonians\cc{ref:Hastings2007-1DAL, ref:Arad2012-1DFFAL,
ref:Arad2013-1DAL, ref:Kuwajara2020-long-range-AL,
ref:Anshu2021-2DAL}. In both cases, the proof follows the intuitive
logic in which a fast convergence to the steady state (or ground
state) can yield a bound on its entanglement. Indeed, given a region
$A$ in the lattice, we can prepare the system in a product state of
this region with the rest of the system and then drive the system
towards its steady state. If the convergence is quick and the
underlying dynamic is local, not too much entanglement is created,
which bounds the entanglement in the steady state. 

There is, however, a highly non-trivial caveat in this program. It
is not a-priori clear that there exists a product state with a large
overlap with the fixed point. If the overlap is exponentially small,
it might take a long (polynomial) time for the dynamics to converge,
even if the convergence is exponentially fast. This is the hard step
in proving an area-law for ground states of gapped local
Hamiltonians.  In Hastings' 1D area-law
proof\cc{ref:Hastings2007-1DAL}, the initial state is $\rho_A\otimes
\rho_{A^c}$, where $\rho_A$, $\rho_{A^c}$ are the reduced density
matrices of the ground state on the regions $A, A^c$. Then a
non-trivial overlap with $\rho_A\otimes \rho_{A^c}$ is shown using
an ingenious argument about the saturation of mutual information.
In Refs.~\cc{ref:Arad2012-1DFFAL, ref:Arad2013-1DAL,
ref:Kuwajara2020-long-range-AL, ref:Anshu2021-2DAL} it is done by
constructing an approximate ground state projector (AGSP) using a
low-degree polynomial of the Hamiltonian. If the Schmidt rank of
this AGSP times its approximation error is smaller than unity, we
are promised that there exists some product state
$\ket{A}\otimes\ket{A^c}$ with a large overlap with the groundstate.

In the open system area-law proof of Refs.\cc{DOC, Brandao} there is
no parallel argument to lower-bound the overlap between the initial
product state and the steady state. In addition, the convergence to
the fixed point is always via the natural $e^{t\mcL}$ map --- which
might not be the most efficient one (in terms of the amount of
entanglement that is generated). In \cRef{DOC} a worse-case overlap
is assumed, which leads to the $\log(\norm{\sigma^{-1}})$ factor in
\eqref{eq:LS-AL-bound}.\footnote{$\norm{\sigma^{-1}}$ is the smallest
possible overlap of full-supported $\sigma$ with another state.} In
\cRef{Brandao}, it is \emph{assumed} that fast convergence is
independent of the initial state. This
assumption, together with additional local fixed-point uniqueness
assumptions that are also made in that work, imply that local
expectation values in the fixed point can be efficiently calculated
by a classical computer, which might make it less interesting from a
computational point of view. In that respect, we believe that our
Lindbladian $\to$ Hamiltonian mapping paves the way for a more
fine-grained analysis of the problem.  Finally, to best of
our knowledge, it has only been shown that general gapped,
detailed-balance Lindbladians satisfy \emph{polynomial} mixing
time\cc{logsobolev}.  For such Lindbladians, rapid mixing has only
been shown in cases where the steady state is a Gibbs state of a
commuting local Hamiltonian\cc{bardet2021rapid,bardet2021entropy}
(which is not true in general).

{~}

The structure of the paper is as follows. In
\autoref{sec:background} we introduce the notation and background of
the paper. In \autoref{sec:results} we state and prove our main
results. In \autoref{sec4} we describe two nontrivial models that satisfy
the assumptions we make, and present explicit forms for their
super-Hamiltonian. We explicitly derive resultant local Hamiltonians
and show it is indeed a local (or exponentially local). We then use
Hamiltonian complexity techniques from Refs.\cc{knabe, Lemm} to
demonstrate a finite gap in those Hamiltonians (and correspondingly
in the original $\mcL$) for a specific example. In \autoref{App:A}
we prove a lemma used in \autoref{thm:quasi-local} regarding square
root of sparse matrices. In \autoref{App:toy1}, \autoref{App:toy2} we
prove that our examples  indeed satisfy all the requirements and
explicitly derive the super-Hamiltonian.

%%%%%%%%%%%%%%%%%%%%%%%%%%%%%%%%%%%%%%%%%%%%%%%%%%%%%%%%%%%%%%%%%%%%%
\section{Background}
\label{sec:background}

%============================= Setup =============================
\subsection{Setup}
\label{sec:setup}

Throughout of this manuscript we use the big-O notation of computer
science. If $n$ is the asymptotic variable, then $X=\bigO{Y}$ means
that there exists $C>0$ such that for sufficiently large $n$,
$X(n)\le C \cdot Y(n)$. On the other hand, $X=\bOmega{Y}$ means that
there exists a constant $c>0$ such that $X(n)\ge c\cdot Y(n)$ for
sufficiently large $n$. Finally, $X=\bTheta{Y}$ means that there
exist constants $C>0,c>0$ such that $c\cdot Y(N) \le X(n) \le C\cdot
Y(n)$ for suffciently large $n$.

We consider many-body systems composed of sites with 
$d$-dimensional local Hilbert space (`qudits', which could
physically correspond to spins, fermions, or hard core bosons) that
reside on a lattice $\Lambda$ of $n$ sites and fixed spatial dimension
(1D, 2D or 3D). The $d$-dimensional local Hilbert space at site
$x\in \Lambda$ is denoted by $\BBH_x$, so that the global Hilbert
space is
\begin{align*}
  \BBH=\bigotimes_{x\in\Lambda} \BBH_x, \qquad 
  N\EqDef \dim(\BBH) = d^{|\Lambda|}.
\end{align*}
We denote the space of linear operators on $\BBH$ by
$L(\BBH)$, which possess a tensor-product structure as well,
i.e., $L(\BBH)=\bigotimes_{x\in \Lambda} L(\BBH_x)$.
$L(\BBH)$ is by itself a finite-dimensional Hilbert-space with
respect to the Hilbert-Schmidt inner product 
\begin{align}
  \braket{O_1,O_2} \EqDef \Tr(O_1^\dagger O_2) ,
    \qquad O_1,O_2\in L(\BBH). 
\end{align} 
We denote the $p$th Schatten norm of an operator $X$ by
$\norm{X}_p\EqDef \Tr(|X|^p)^{1/p}$. The $p=1$ case is the trace
norm, commonly used to measure distance between density matrices.
The $p=2$ case is the norm derived from the Hilbert-Schmidt
inner-product. The $p=\infty$ norm, defined by the maximal singular
value of $X$, is the operator norm, and is denoted in this paper by
$\norm{X}$. 

Given a subset of the lattice $S\subseteq \Lambda$, we denote its
local Hilbert space by $\BBH_S$, i.e., $\BBH_S=\bigotimes_{x\in S}
\BBH_x$. An operator $O\in L(\BBH)$ is \emph{supported} on a subset
$S\subseteq \Lambda$ if it can be written as $O=O_S\otimes
\Id_{\BBH_{S^c}}$ where $O_S\in L(\BBH_S)$ and $\Id_{\BBH_{S^c}}$ is
the identity operator on the complementary Hilbert space. 

Linear operators acting on the operators space $L(\BBH)$ are called
\emph{superoperators}, and will usually be denoted by curly letters
(e.g $\mcL,\mcP$). As in the case of operators, we say a
superoperator $\mcB$ is supported on a subset $S\subseteq \Lambda$
of the sites in the lattice if it can be written as $\mcB_S \otimes
\mcI_{S^c}$, where $\mcB_S$ is a superoperator on $L(\BBH_S)$, and
$\mcI_{S^c}$ is the identity operation on the complementary
operators space.

Given a superoperator $\mcB$, its dual map $\mcB^*$ is the unique
map that satisfies
\begin{align}
\label{eq:self_adjoint}
    \braket{O_1,\mcB(O_2)}= \braket{\mcB^*(O_1), O_2}
      \tab \forall O_1,O_2\in L(\BBH),
\end{align}
i.e., the adjoint under the Hilbert-Schmidt inner-product.  It is
easy to check that the adjoint of the superoperator 
$\mcB(\rho)=\sum_{i,j} c_{ij} F_i \rho F_j^\dagger$ is given by
$\mcB^*(\rho)=\sum_{i,j}c_{ij}^* F_i^\dagger \rho F_j$. We call
$\mcB$ \emph{self-adjoint} if $\mcB=\mcB^*$.

%====================================================================
\subsection{Open quantum systems}
\label{sec:open-systems}

In this section we provide some basic definitions and results about
Markovian open quantum systems. For a detailed introduction to this
subject we refer the reader to Refs.\cc{ref:Breuer2006-opensys,
wolf}.

We study Markovian open quantum systems governed by a
time-independent \emph{Lindbladian} (also known as a
\emph{Liouvillian}) $\mcL$ in the Schr\"odinger picture.  Formally,
this means that the quantum state describing the system evolves
continuously by a family of \emph{completely positive trace
preserving} (CPTP) maps $\{\mcT_t\}_{t\ge 0}$, which are given by
$\mcT_t \EqDef e^{\mcL t}$ so that
\begin{align}
  \rho(t) = \mcT_t \rho_0 \quad\Leftrightarrow\quad 
    \partial_t \rho(t) = \mcL\rho(t) , \quad \rho(0) = \rho_0 .
\end{align}
$\{\mcT_t\}$ is known in the literature as a \emph{dynamical
semigroup} or \emph{Quantum Markovian Semigroup} due to the identity
$\mcT_{t}\circ \mcT_s = \mcT_{t+s}$.  A necessary and sufficient
condition for $\mcL$ to be the generator of a semigroup is given by
the following theorem (see, for example, Theorem 7.1 in
\cRef{wolf}):
\begin{theorem}
\label{thm:Lin-jump}
  $\mcT_t = e^{\mcL t}$ is a dynamical semigroup iff it can be
  written as
  \begin{align}
  \label{eq:jump-form}
    \mcL(A) = -i[H,A] + \sum_j L_j A L_j^\dagger 
      -\frac{1}{2}\{ L_j^\dagger L_j, A\} ,
  \end{align}
  where $H$ is a Hermitian operator and $L_j$ are operators.
\end{theorem}
The operator $H$ is known as the Hamiltonian of the system, and it
governs the coherent  part of the evolution. The operators $L_j$
are often called \emph{jump operators}, and are
responsible for the dissipative part of the evolution. 

Given a Lindbladian $\mcL$, the representation in 
\eqref{eq:jump-form} is not unique; for example, we can always
change $H\to H+c\Id$ for $c\Id$ for $c\in \mathbb{R}$ without
changing $\mcL$. In addition, $\mcL$ will remain the same under the
transformation $L_j\to L_j+c_j\Id$ and $H\to
H+\frac{i}{2}\sum_j(c_j^*L_j - c_jL_j^\dagger)$. Therefore, we can
assume without loss of generality that there is a representation of
$\mcL$ in which both $H$ and $L_j$ are traceless.

This condition, however, does not fully fix the jump operators, and
there can be several traceless jump operators representations of the
same Lindbladian. The following theorem, which is an adaptation of 
Proposition~7.4 in \cRef{wolf}, shows how these representations are
related.
\begin{theorem}[Freedom in the representation of $\mcL$, 
Proposition~7.4 in \cRef{wolf}] \ \\
\label{thm:freedom}
  Let $\mcL$ be given by \Eq{eq:jump-form} with traceless $\{L_j\} $
  and $H$. If it can also be written using traceless
  $\{L'_i\}$ and $H '$, then $H=H'$ and there exists a unitary matrix
  $U$ such that 
  $L'_i = \sum_j U_{ij} L_j$,
  where the smaller
  set of jump operators is padded with zero operators.
\end{theorem}

The following are two important properties of Lindbladians that
will be used extensively in this work. First, we define the notion
of locality of a Lindbladian (which can be generalized to a
general super-operator).
\begin{definition}[$k$-body Lindbladians]
\label{def:k-body-L} 
  We say that $\mcL$ is a \emph{$k$-body} Lindbladian if it can be
  written as in \eqref{eq:jump-form} with jump operators $L_j$ that
  are supported on at most $k$ sites, and in addition the
  Hamiltonian $H$ can be written as $H=\sum_i h_i$, with every $h_i$
  also supported on at most $k$ sites.  We say that $\mcL$ is a
  \emph{geometrically local} $k$-body Lindbladian if, in addition
  to being $k$-body, every $L_j$ and $h_i$ are supported on
  neighboring lattice sites.
\end{definition}
Note that we are using the name "k-body" Lindbladian instead of
``$k$-local'': as it is often done in the Hamiltonian complexity
literature\cc{ref:Kitaev2002-QCbook}, the term $k$-local is reserved
to local terms supported on at most $k$ \emph{qubits}. Here, we are
allowing any constant local dimension $d$. 

Next, we define the spectral gap of a Lindbladian.
\begin{definition}[Spectral gap of a Lindbladian]
\label{def:gap-L} 
  The \emph{spectral gap} of the Lindbladian is
  defined as the minimal real part of its non-zero eigenvalues
  \begin{align}
    \label{def:gap}
    \gap(\mcL) \EqDef \min_{0\neq\lambda 
      \in \mathrm{Spec}(\mcL)} |\Re(\lambda)| .
  \end{align}
\end{definition}
The spectral gap of a Lindbladian controls the \emph{asymptotic}
convergence rate of the dynamics to a steady state\cc{gaps_Znidari},
though a finite gap by itself does not guarantee short-time
convergence\cc{cutoff}.
  
We conclude this section by listing few well-known facts about
Lindbladians. We refer the reader to chapters 6,7 of \cRef{wolf} for
proofs and details:
\begin{description}  
  \item[Fact 1:] $\mcL$ is an  hermicity-preserving superoperator:  
    $\left(\mcL(A)\right)^\dagger = \mcL(A^\dagger)$, and the same
    holds for $\mcL^*$.
    
  \item[Fact 2:] There is always at least one quantum state $\sigma$ that satisfies 
    $\mcL(\sigma)=0$, namely $\sigma$ is a fixed point of the
    time-evolution. We refer to it as a \emph{steady state}.
        
  \item[Fact 3:] $\mcT_t$ is a contractive map and consequently, 
    $\mcL$ has only non-positive real parts in its spectrum
    (Proposition 6.1 in \cRef{wolf}).
    
\end{description}

%=========================== detailed-balance ============================
\subsection{Quantum detailed-balance}
\label{sec:QDB}

In this work, we follow Refs.\cc{ALICKI1976249,Maas,Fangola} in
defining the detailed-balance condition for the Lindbladian system.
We begin by describing classical detailed-balance, and
then use it to define the corresponding quantum condition.

Classically, let $P\in \mathbb{M}_n$ be the transition matrix of a
Markov chain over the discrete state of states $\{1, 2, \ldots,n\}$,
i.e., $P_{ij} \EqDef \Prob(j\to i|j)$, and let $\pi$ denote a
probability distribution on these states.  Then $P$ is said to
satisfy the detailed-balance condition with respect to a fully
supported $\pi$ (i.e., $\pi_i>0$ for all $i\in [n]$) if the
probability of observing a $i\to j$ transition is identical to the
probability of observing a $j\to i$ transition, when the system
state is described by $\pi$. Mathematically, this means $P_{ij}\pi_j
= P_{ji}\pi_i$.  This condition implies that $\pi$ is a steady state
of the Markov chain, however, the converse is not always
true\cc{Sen_book}.

We can also write this condition in terms of matrices. Defining
$\Gamma_\pi$ to be the diagonal matrix $(\Gamma_\pi)_{ij} =
\pi_i\delta_{ij}$, the detailed-balance condition can also be
written as the matrix equality:
\begin{align}
\label{eq:CDB1}
  P\Gamma_\pi = \Gamma_\pi P^T .
\end{align}
Alternatively, noting that $\Gamma_\pi$ is a positive definite
matrix, the above condition is equivalent to the condition of
$\Gamma_\pi^{-1/2} P \Gamma_\pi^{1/2}$ being symmetric:
\begin{align}
\label{eq:CDB2}
  \Gamma_\pi^{-1/2} P \Gamma_\pi^{1/2}
    = \Gamma_\pi^{1/2} P^T \Gamma_\pi^{-1/2}
    = (\Gamma_\pi^{-1/2} P \Gamma_\pi^{1/2})^T .
\end{align}

These two conditions can be generalized to the quantum setting by
changing $P\to \mcL$ and $\pi\to\sigma$ for some reference quantum
state $\sigma$. Just as in the classical case, we demand that
$\sigma$ is invertible, i.e., it has no vanishing eigenvalues. The
quantum analog of \Eq{eq:CDB1} should be an equation over
superoperators that act on quantum density operators. While $P$ is
replaced by $\mcL$, $\Gamma_\pi$ should be replaced by a
superoperator that multiplies an input state by the quantum state
$\sigma$. But as $\sigma$ does not commute with all quantum states
(unless it is the completely mixed state), there are several ways to
define this multiplication. In particular, for every $s\in [0,1]$,
we might define the ``multiplication by $\sigma$''
superoperator\footnote{Here we do not use a curly letter to denote
the superoperator $\Gamma_s$ in order to be consistent with previous
works.}
\begin{align}
  \label{def:Gamma-s}
  \Gamma_s(A) &\EqDef \sigma^{1-s} A \sigma^s .
\end{align}
It is easy to verify that, just as in the classical case, the
superoperator $\Gamma_s$ is self adjoint: $\Gamma_s = \Gamma_s^*$.
Moreover, $\Gamma_s$ is invertible, and for every $x\in
\mathbb{R}$, we have $\Gamma^x_s(A) = \sigma^{x(1-s)} A
\sigma^{xs}$.

With this notation, the quantum detailed balance (QDB) condition is
defined as the following generalization of
the classical condition \eqref{eq:CDB1}:
\begin{definition}[Quantum detailed balance]
  \label{def:QDB} A Lindbladian satisfies quantum detailed-balance
  with respect to some invertible (full-rank) state $\sigma\in
  L(\BBH)$ and $s\in [0,1]$ if
  \begin{align}\label{eq:QDB}
    \mcL \circ \Gamma_s = \Gamma_s \circ \mcL^*, 
  \end{align}
  where $\Gamma_s$ is the superoperator defined in
  \eqref{def:Gamma-s}.
\end{definition}
We note that not every steady-state of a Lindbladian defines
a superoperator $\Gamma_s$ with respect to which the Lindbladian
obeys detailed-balance\cc{laracuente}.

As in the classical case, the quantum detailed-balance condition can
be formulated in a few equivalent ways:
\begin{claim}
\label{cl:QDB}
  Given a Lindbladian $\mcL$, an invertible state $\sigma\in
  L(\BBH)$, and $s\in[0,1]$, the following conditions are
  equivalent:
  \begin{enumerate}
    \item  (QDB1)
      $\mcL$ satisfies quantum detailed balance with respect to
      $\sigma$ for some $s$.
    
    \item (QDB2) \label{def:QDB2}
      The superoperator $\Gamma^{-1/2}_s \circ \mcL \circ
      \Gamma_s^{1/2}$ is self-adjoint, namely $ \Gamma^{-1/2}_s
      \circ \mcL \circ \Gamma_s^{1/2}=\Gamma_s^{1/2} \circ
      \mcL^*\circ \Gamma^{-1/2}_s$.

    \item (QDB3) \label{def:QDB3}
      $\mcL^*$ is self-adjoint with respect to the inner-product
      defined by 
      \begin{align}
      \label{def:inner-s}
        \braket{A,B}_s\EqDef \Tr(A^\dagger \cdot \Gamma_s(B))
          = \Tr(A^\dagger\cdot\sigma^{1-s}B\sigma^s).
      \end{align}

  \end{enumerate}
\end{claim}
\begin{proof}
  (QDB2) is equivalent to (QDB1) by a simple conjugation with
  $\Gamma^{1/2}$, and (QDB3) is equivalent to (QDB1) by
  \begin{align*}
  \braket{\mcL^*(A),B}_s &= \Tr\Big[\big(\mcL^*(A)\big)^\dagger\cdot
    \Gamma_s(B)\Big]
     = \Tr\big[A^\dagger\cdot \mcL\circ\Gamma_s(B)\big] \\
     &= \Tr\big[ A^\dagger\cdot \Gamma_s\circ \mcL^*(B)\big]
      = \braket{A, \mcL^*(B)}_s .
  \end{align*}
\end{proof}

The $s=1$ is commonly known as the Gelfand-Naimark-Segal
(GNS) case, and its inner product $\braket{A,B}_1 = \Tr(\sigma A^\dagger
B)$ is often called the GNS inner-product. Similarly, the $s=1/2$
case is called the Kubo-Martin-Schwinger (KMS) case, with
$\braket{A,B}_{1/2}$ known as the KMS inner-product.

 It was shown in
\cRef{Maas} (see Lemmas 2.5, 2.8 therein) that any superoperator
that is self adjoint with respect to the inner product defined by
some $s\in [0,1/2)\cup (1/2, 1]$ is self adjoint with respect to the
inner product defined by all $s^\prime\in [0,1]$, including
$s^\prime=1/2$. Therefore by QDB3, if $\mcL$ satisfies the
detailed-balance condition for some $s\in [0,1/2)\cup (1/2, 1]$, then it
satisfies it for all other $s\in [0,1]$. 

Finally, note that if $\mcL$ satisfies detailed balance with respect
to $\sigma$, then $\sigma$ is automatically a steady state, since
$\Gamma^{-1}_s(\sigma) = \Id$ and therefore
\begin{align*}
  \mcL(\sigma)=\Gamma_s \circ \mcL^* \circ \Gamma^{-1}_s(\sigma)
    = \Gamma_s\circ \mcL^*(\Id)=0 .
\end{align*}

We refer the reader to Refs.~\cc{bolanos2013infinite, DOC,
chi,szehr2015spectral} for other definitions and generalization to
quantum detailed-balance.

%====================================================================
\subsection{The Canonical Form}
\label{sec:canonical-form} 

The QDB condition has well-known implications to the structure of
the Lindbladian $\mcL$. In this section we describe some of the
central consequences of this condition, and in particular the
so-called \emph{canonical form}, which will be used later. None of
the results in this section are new, as they already appeared in
several works (see, for example,
Refs.\cc{Maas,ALICKI1976249,Gorini}). Nevertheless, we repeat some
of the easy proofs for sake of completeness.

Our starting point is the \emph{modular superoperator}, which is
central for constructing a canonical representation of
detail-balanced Lindbladians. 
\begin{definition}[The modular superoperator]
  Given an invertible state $\sigma$, its associated modular
  superoperator is defined by 
  \begin{align}
  \label{def:modular_operator}
    \mcG_\sigma(A) \EqDef \sigma A \sigma^{-1} .
  \end{align}
  We note that we are using a
  greek letter $\Delta$ to denote the modular superoperator instead
  of curly letter. This is done for being consistent with the
  notations of Refs.\cc{Maas, Szegedy}.
\end{definition}
Below are few properties of $\mcG_\sigma$, which follow almost
directly from its definition.
\begin{claim} 
  The modular superoperator $\mcG_\sigma$ has the following
  properties:
  \begin{enumerate}
    \item $\mcG^{-1}_\sigma(A) = \sigma^{-1} A \sigma$ and 
      $\big[\mcG_\sigma(A)\big]^\dagger = \mcG^{-1}_\sigma(A^\dagger)$.
    \item Self-adjointness: $\mcG_\sigma^* = \mcG_\sigma$.
    \item Positivity: $\braket{A, \mcG_\sigma(A)} > 0$ for all non-zero
      operators $A$.
  \end{enumerate}
\end{claim}
\begin{proof}
  Properties 1 and 2 follow from definition. For 3, note that
  \begin{align*}
    \braket{A, \mcG_\sigma(A)} = \Tr(A^\dagger \sigma A\sigma^{-1}) =
      \Tr\big[ (\sigma^{1/2} A \sigma^{-1/2})^\dagger \cdot 
        (\sigma^{1/2} A \sigma^{-1/2})\big]
  \end{align*}
  if $A\neq 0$, then by the invertability of $\sigma$ it follows
  that also $\sigma^{1/2} A \sigma^{-1/2} \neq 0$, hence the RHS
  above is positive.
\end{proof}
Properties 2 and 3 imply that $\mcG_\sigma$ is fully diagonalizable
by an orthonormal eigenbasis of operators $\{S_\alpha\}$ and
positive eigenvalues $\{e^{-\omega_\alpha} \}$, where $\alpha=0, 1,2,\ldots, N^2-1$ is a running index. 
As $\mcG_\sigma(\Id) = \Id$, we can fix $S_0 =
\frac{1}{\sqrt{N}}\Id$, $\omega_0=0$, and conclude that $\Tr(S_\alpha)
= \sqrt{N}\braket{S_0, S_\alpha} = 0$ for every $\alpha>0$. Finally, from
property~(1) we find that 
\begin{align*}
  \mcG_\sigma(S^\dagger_\alpha) = \big[
  \mcG^{-1}_\sigma(S_\alpha) \big]^\dagger
   = \big[e^{\omega_\alpha} S_\alpha]^\dagger 
     = e^{\omega_\alpha} S^\dagger_\alpha .
\end{align*}
Therefore, $S^\dagger_\alpha$ is also an eigenoperator of
$\mcG_\alpha$ with eigenvalue $e^{\omega_\alpha}$. All of these
properties are summarized in the following corollary, which defines
the notion of a \emph{modular basis}.
\begin{corol}[Modular basis]
\label{corol:modular-basis} Given an invertible state $\sigma$, the
  modular superoperator $\mcG_\sigma$ has an orthonormal
  diagonalizing basis $\{S_\alpha\}$ with $\alpha=0, 1, \ldots,
  N^2-1$ and the following properties:
  \begin{enumerate}
    \item $\braket{S_\alpha, S_\beta} = \delta_{\alpha\beta}$
    \item $S_0=\frac{1}{\sqrt{N}}\Id$ and $\Tr(S_\alpha)=0$ for
      $\alpha>0$.
    \item $\mcG_\sigma(S_\alpha) = e^{-\omega_\alpha} S_\alpha$ and
      $\mcG_\sigma(S^\dagger_\alpha) = e^{\omega_\alpha}
      S^\dagger_\alpha$.
    \item For every $\alpha$ there exists $\alpha'$ such that 
      $S_\alpha^\dagger = S_{\alpha'}$.
  \end{enumerate}
  The basis $\{S_\alpha\}$ is called a modular basis.
\end{corol}
Being a diagonalizing basis, the modular basis is unique up to
unitary transformations within every eigenspace. The numbers
$\omega_\alpha$, which determine the eigenvalues of the modular
superoperator, are called \emph{Bohr frequencies}. We say that an
operator $O$ has well-defined Bohr frequency $\omega$ if
$\mcG_\sigma(O) = e^{-\omega}O$, i.e., it belongs to the eigenspace
of $\mcG_\sigma$ with eigenvalue $e^{-\omega}$. For example, it is
evident that $\sigma$ itself (or any other state that commutes with
it) has a well-defined Bohr frequency $\omega=0$ since
$\Delta_\sigma(\sigma)=\sigma$.

The Bohr frequencies are related to the eigenvalues of $\sigma$. To
see this, we use an explicit construction of a modular basis. Given
the spectral decomposition $\sigma=\sum_i e^{-E_i} \ketbra{i}{i}$,
we define the operators $O_{ij} \EqDef \ketbra{i}{j}$ with $i,j \in
1, \ldots, N$, and note that they form an orthonormal diagonalizing
basis of $\mcG_\sigma$:
\begin{align*}
  \mcG_\sigma(O_{ij}) = \sigma \ketbra{i}{j} \sigma^{-1} 
    = e^{-\omega_{ij}} O_{ij}, \qquad \omega_{ij} 
  \EqDef E_i - E_j .
\end{align*}

Finally,
Lindbladians that satisfy the QDB condition for $s\in
[0,1/2)\cup(1/2,1]$ with respect to an invertible state $\sigma$
can be written in a particular canonical way. This was first proved
in \cRef{ALICKI1976249} under slightly different conditions. Here we
will follow \cRef{Maas}, and use an adapted form of Theorem 3.1 from
that reference: 
\begin{theorem}[Canonical form of QDB Lindbladians, adapted from
Theorem 3.1 in \cRef{Maas}] 
\label{thm:canonical_form}
  A Lindbladian $\mcL$ satisfies quantum detailed-balance with
  respect to an invertiable state $\sigma$ and $s\in
  [0,1/2)\cup(1/2,1]$ if and only if it can be written as
  \begin{align} 
  \label{def:canonical_form}
    \mcL(A) = \sum_{\alpha\in I} \gamma_\alpha e^{-\omega_\alpha /2}
      \Big( S_\alpha A S_\alpha^\dagger 
        - \frac{1}{2}\big\{S_\alpha^\dagger S_\alpha, A\big\}\Big)
  \end{align} 
  where $\omega_\alpha$ are Bohr frequencies, the jump operators
  $\{S_\alpha\}$ are taken from a modular basis, and  
  $\{\gamma_\alpha\}$
  are positive weights that satisfy\footnote{Recall that in
  accordance with the modular basis definition (see
  Corollary~\ref{corol:modular-basis}), for every index $\alpha$,
  there exists an index $\alpha'$ such that $S^\dagger_\alpha =
  S_{\alpha'}$ and $\omega_\alpha=-\omega_{\alpha'}$.}
  $\gamma_\alpha=\gamma_{\alpha '}$, and in particular the set of
  indices $I$ contains $\alpha '$ for every $\alpha\in I$.
\end{theorem} 

{~}

We end this section with a few remarks.
\begin{enumerate}
  \item In the original text of \cRef{Maas}, the formula for $\mcL$
    is given in the Heisenberg picture, which is the formula for
    $\mcL^*$ in our notation. Additionally, the $\gamma_\alpha$
    weights are missing, as they are instead absorbed into the
    $S_\alpha$. The normalization condition $\braket{S_\alpha, S_\beta}
    = \delta_{\alpha\beta}$, however, remained unchanged, which we
    believe is a mistake.
    
  \item In \Eq{def:canonical_form} we sum over strictly 
    positive $\gamma_\alpha$ and ignore the vanishing weights. In
    \cRef{Maas} (Equation 3.3) the authors considered the vanishing
    coefficients as well in the summation.
        
  \item It follows from the canonical representation that if $\mcL$
    satisfies the QDB condition, it is purely dissipative, i.e., its
    Hamiltonian is vanishing. In the literature, QDB condition is often referred to the dissipator part of the Lindbladian.
        
  \item \autoref{thm:canonical_form} applies to any Lindbladian defined 
      on a finite dimensional Hilbert-space regardless of the
      many-body structure of the underlying Hilbert space, which can
      describe spins, fermionic, or hard-core  bosonic systems. 
    
  \item A physical example of a Lindbladian that
    satisfies the QDB condition (for $s=1$) is the Davies generator
    of some Hamiltonian $H$\cc{davies}. It is the principle example of a
    semigroup whose unique fixed point is a Gibbs state (i.e.
    thermal state), and is often referred to as
    a thermal semigroup\cc{logsobolev, kastoryano2016quantum,
    bardet2021rapid}. 
\end{enumerate}

%%%%%%%%%%%%%%%%%%%%%%%%%%%%   Results.  %%%%%%%%%%%%%%%%%%%%%%%%%%%%
\section{Main Results}
\label{sec:results} 

In this section we give precise statements of our main results about
the mapping of detailed-balanced local Lindbladians to local
Hamiltonians (\autoref{sec:mapping}), and discuss their application to the
complexity of the steady states of these systems
(\autoref{sec:localapp}). The proofs of the main results are given in
\autoref{sec:Proofs}.

%================================================================== 
\subsection{Mapping detailed-balanced Linbaldians to local Hamiltonians}
\label{sec:mapping}

 We consider a geometrically-local $k$-body Lindbladian $\mcL$ 
defined on a finite $D$ dimensional lattice $\Lambda$ of qudits with
local dimension $d$. We assume that $\mcL$ satisfies the QDB
condition in \Def{def:QDB} with respect to some $s\in [0,1/2)\cup
(1/2,1]$ and a unique, full-rank steady state $\sigma$.
As $\mcL$
is detailed-balanced, it does not have a Hamiltonian part, and therefore
it can be written as 
\begin{align}
\label{eq:Lindbladian2}
  \mcL(\rho) = \sum_j L_j\rho L_j^\dagger 
    - \frac{1}{2}\{L_j^\dagger L_j,\rho\},
\end{align}
with $k$-body traceless jump operators $L_j$. We also assume that
$\mcL$ has a spectral gap $\gamma>0$. Note that since $\mcL$ obeys
detailed-balance, it has a real spectrum, and $\gap(\mcL)=\gamma$.

Our goal is to map $\mcL$ and its steady state $\sigma$ to a local
Hamiltonian problem, where they can be analyzed using plethora of
well-established tools\cc{knabe, ref:Hastings2007-1DAL,
ref:Arad2013-1DAL, ref:Landau2015-1Dalg, hastings2021gapped}. This
is achieved by mapping the Lindbladian to the super-operator
\begin{align}
\label{def:super-H}
  \mcH \EqDef -\Gamma_s^{-1/2}\circ \mcL \circ \Gamma_s^{1/2} ,
\end{align}
which is self-adjoint according to Claim~\ref{cl:QDB} (QDB2). We
shall refer to $\mcH$ as the \emph{super-Hamiltonian}, and study its
vectorization as a local Hamiltonian\footnote{To obtain a physical
interpretation for the super Hamiltonian, one can relate its
imaginary time evolution to the Lindblad evolution by similarity
transformation, namely, $e^{-\mcH t}= \Gamma_s^{-1/2}\circ e^{\mcL
t} \circ \Gamma_s^{1/2}$.}. 

Since $\mcH$ and $\mcL$ are related by a similarity transformation,
they share the same spectrum (up to an overall global minus sign),
and therefore also $\mcH$ has a spectral gap $\gamma>0$. Moreover,
it is easy to see that an eigenoperator $A$ of $\mcL$ maps to an
eigenoperator $\Gamma^{-1/2}_s(A) = \sigma^{-(1-s)/2}A\sigma^{-s/2}$
of $\mcH$, and in particular the steady state $\sigma$ maps to
$\sqrt{\sigma}$ which is in the kernel of $\mcH$.  We remark that
the superoperator $\Gamma_s^{-1/2}\circ \mcL \circ \Gamma_s^{1/2}$
and its steady state $\sqrt{\sigma}$ were already studied in
\cRef{chi,DOC,alicki2009} as tools for relating the $\chi^2$ decay
constant to the spectral gap of a QDB Lindbladian, and in
\cRef{kastoryano2016quantum} for demonstrating strong clustering of
information using the detectability lemma. While the locality
of the super-Hamiltonian was not addressed in \cRef{chi,DOC}, it was
used in \cRef{kastoryano2016quantum, alicki2009}.  However, in these
papers, locality was automatically achieved since the steady state
was a Gibbs state of a commuting local Hamiltonian. Our work is
devoted to proving this in the general case.

We are left with the task of showing that $\mcH$ is geometrically
local, or at least local with decaying interactions. We will prove
this under two possible assumptions. In \autoref{thm:local} we
assume a certain linear-independence conditions on the jump
operators, and consequently find that $\mcH$ is a $k$-body
geometrically local. In \autoref{thm:quasi-local} we expand the jump
operators in terms of a local orthonormal basis, and show that if
the coefficient matrix in that basis is gapped, then $\mcH$ is
$2k$-body local with coefficients that decay exponentially with the
lattice distance.

%--------------------- Theorem local ------------------
\begin{theorem}
\label{thm:local}
  Let $\{L_j\}$ be a set of linearly independent and normalized jump
  operators such that for every index $j$, there exists an index
  $\pi(j)$ such that $L_j^\dagger = L_{\pi(j)}$. Assume that $\mcL$
  satisfies the QDB condition (\Def{def:QDB}) for $s\in
  [0,1/2)\cup(1/2,1]$ and is given by
  \begin{align}
  \label{eq:local-indp-L}
    \mcL(\rho) = \sum_j c_j \big(L_j\rho L_j^\dagger 
      - \frac{1}{2}\{ L_j^\dagger L_j, \rho\}\big)
  \end{align}
  for some positive coefficients $c_j>0$. Then the super-Hamiltonian
  \eqref{def:super-H} is given by
  \begin{align} 
  \label{eq:super-H-1}
    \mcH(\rho) = -\sum_j\big(\sqrt{c_j c_{\pi(j)}} 
      L_j\rho L_j^\dagger  
        - \frac{c_j}{2}\{L_j^\dagger L_j,\rho\}\big) .
  \end{align}
  Consequently, if $\mcL$ is $k$-body and geometrically local, then
  so is $\mcH$.
\end{theorem}

%----------------- Theorem Quasi-local  ------------------

To state the second theorem, we first expand our jump operators in
\Eq{eq:Lindbladian2} in terms of an orthonormal operators basis
$\{P_a\}$, which we assume to be Hermitian and $k$-body
geometrically local: 
\begin{align}
\label{eq:P-expansion}
  L_j \EqDef \sum_a R_{aj} P_a .
\end{align}
For example, if the local Hilbert dimension is $d=2$, we can take
$\{P_a\}$ to be products of $k$ Paulis on neighboring sites.
Substituting \Eq{eq:P-expansion} in \Eq{eq:Lindbladian2} yields
\begin{align}
  \label{eq:LP}
  \mcL(\rho) = \sum_{a,b} C_{ab} \big(P_a \rho P_b 
    -\frac{1}{2}\{P_bP_a, \rho\}\big) ,
\end{align}
where $C_{ab} = \sum_j R_{aj} R^*_{bj}$ is a positive
semi-definite (and hence Hermitian) matrix. With this notation, we
have
\begin{theorem} 
\label{thm:quasi-local}
  Let $\mcL$ be a Lindbladinan given by \Eq{eq:LP} with $\{P_a\}$
  being a set of orthonormal hermitian operators, and assume that
  $\mcL$ satisfies the QDB condition (\Def{def:QDB}) with respect to
  $s\in[0,1/2)\cup(1/2,1]$.  Then  the
  super-Hamiltonian \eqref{def:super-H} is given by
  \begin{align}
  \label{eq:H-in-P}
      \mcH(\rho) = -\sum_{a,b}\left( (C\cdot C^*)^{1/2}_{ab}\,
        P_a \rho P_b - \frac{1}{2} C_{ab} \{P_bP_a, \rho\}\right) 
  \end{align}
where $C^*$ is the complex conjugate of $C$.
\end{theorem}
We note as a side-product of the proof, $C, C^*$ commute, and as
they are both non-negative matrices, the sqaure root $(C\cdot
C^*)^{1/2}$ is well-defined.  As a corollary, we achieve a
sufficient condition for quantum detailed balance, which can be
easily checked with a polynomial computation \emph{without
the knowledge of the steady state $\sigma$}. Let us formally state
this result: 
\begin{corol}[A necessary condition for quantum detailed-balance]
\label{corol:necessary-cond-QDB}
    Let $\mcL$ be a Lindbladinan given by \Eq{eq:LP} with $\{P_a\}$
  being a set of orthonormal hermitian operators, and assume that
  $\mcL$ satisfies the QDB condition (\Def{def:QDB}) with respect to
  $s\in[0,1/2)\cup(1/2,1]$.  Then the matrix $C$ from \Eq{eq:LP}
  commutes with its complex conjugate $C^*$.
\end{corol}

Assuming $\mcL$ is $k$-body and geometrically local, the matrix $C$
is sparse: $C_{ab}\neq 0$ only for $a,b$ for which $P_a$ and $P_b$
appear in the expansion of the same jump operator $L_j$, and
therefore their joint support is at most gemotrically $k$-local.
This means that the $\sum_{a,b} C_{ab}\{P_bP_a, \rho\}$ term in
\Eq{eq:H-in-P} is geometrically-local $k$-body. However, the
geometrical locality of the first term $\sum_{a,b}(C\cdot
C^*)^{1/2}_{ab}\, P_a \rho P_b$ is not so clear, as it might mix
$P_a$ and $P_b$ with distant supports (which might lead to a
long-range $\mcH$). In the following lemma, we show that as long as
the smallest non-zero eigenvalue of $C\cdot C^*$ is $\Omega(1)$, the
coefficient $(C\cdot C^*)^{1/2}_{ab}$ decays exponentially with the
distance between the support of $P_a, P_b$, making the terms of
$\mcH$ decay exponentially with distance.

\begin{lemma}
\label{lem:exp-decay-matrix}
  Let $\lambda_{\mathrm{min}}>0$ be the smallest non-zero eigenvalue of
  $C\cdot C^*$, and assume that $|C_{ab}|\le J$ for every $a,b$. 
  Finally, let $|a-b|$ denote the lattice distance between the
  supports of $P_a, P_b$. Then
  \begin{align*}
      |(C\cdot C^*)^{1/2}_{ab}| 
        \le c_1 J e^{-c_2|a-b|\cdot\lambda_{\mathrm{min}}/J^2}, 
  \end{align*}
  where $c_1, c_2$ are constants that depend only on the geometry of
  the lattice and $k$.
\end{lemma}
The proof of this lemma is given in \autoref{App:A}.

\begin{corol} 
\label{cor:quasi-local} 
  Let $\mcL$ be a geometrically local $k$-body Lindbladian given by
  \Eq{eq:LP} that satisfies the QDB condition with respect to a full
  rank state and $s\in [0,1/2)\cup(1/2,1]$. Let
  $\lambda_{\mathrm{min}}$ be the smallest non-zero eigenvalue of
  $C\cdot C^*$, and assume that $|C_{ab}|\le J$ for all $(a,b)$
  pairs. Then the super-Hamiltonian in~\autoref{thm:quasi-local} is
  of the form 
  \begin{align*}
    \mcH = \sum_{a,b}\mcH_{ab},
  \end{align*}
  where $\mcH_{ab}$ is a $2k$-local super operator that acts non
  trivially on $\supp(P_a)\cup\supp(P_b)$, with an exponentially
  decaying interaction strength $\norm{\mcH_{ab}} =
  Je^{-O(|a-b|\lambda_{\mathrm{min}}/J^2)}$.
\end{corol}

We conclude this section with two remarks:
\begin{enumerate}

  \item Related to Corollary~\ref{corol:necessary-cond-QDB}, also the
    linear independence of $\{L_j\}$ from \autoref{thm:local}, as
    well as the minimal non-vanishing eigenvalue of $C\cdot C^*$
    from \autoref{thm:quasi-local}, can be checked and calculated
    efficiently without any prior knowledge of the steady state or
    the gap. This might be beneficial in numerical applications of
    these mappings, as the formula for $\mcH$ is explicitly given in
    terms of the jump operators and $C$.
    
  \item Theorems~\ref{thm:local} and \ref{thm:quasi-local} are valid 
    for any Lindbladians defined on a finite-dimensional algebra of
    operators, as they do not use any commutativity or locality
    properties of the jump operators or basis. Therefore these
    theorems apply to fermionic/hard-core bosonic systems
    as well.
\end{enumerate}
Before providing proofs of Theorems~\ref{thm:local}
and~\ref{thm:quasi-local} in \autoref{sec:Proofs}, we will use the
following subsection to discuss their main consequences.

%================================================================== 
\subsection{Application to the complexity of steady states of QDB
Lindbladians}
\label{sec:localapp}

The mapping $\mcL \mapsto \mcH$ allows us to easily import results
from the Hamiltonian realm into the Lindbladian realm. In the
following subsection, we describe some of the main results that can
be imported and discuss the underlying techniques. We begin by
describing the vectorization mapping that allows us to map
super-Hamiltonians and density operators to Hamiltonians and
vectors, respectively.

%--------------------------------------------------------------
\subsubsection{The vectorization map}

\begin{figure}
  \begin{center}
    \includegraphics[scale=1]{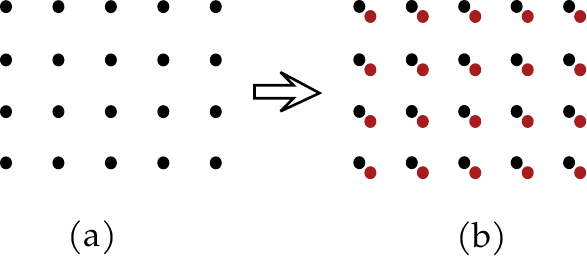}
  \end{center}
  \caption{\label{fig:vectorization} The composite system 
    in the vectorization of a many-body system on a lattice.
    Starting from lattice on which $\sigma$ is defined (the black
    dots in (a)), we introduce
    a `fictitious' site next to every site in the lattice (the red
    dots in (b)). The composite system can then be viewed as a new
    many-body lattice system in which the local dimension is $d^2$.}
\end{figure}

We map the super Hamiltonian \eqref{def:super-H} to a local
Hamiltonian using standard \emph{vectorization}, which is the
well-known isomorphism $L(\BBH) \mapsto \mathbbm{H}\otimes
\mathbbm{H}$ (see, for example, \cRef{ref:Watrous2018-QI}).  We will
denote this mapping by $X \to \dket{X}$, where $X\in L(\BBH)$ is an
operator and $\dket{X}\in \BBH\otimes\BBH$ is a vector. For the sake of
completeness, we explicitly define it here and list some of its
main properties. The map is defined first on the standard basis
elements by
\begin{align*}
  \ketbra{i}{j} \mapsto \ket{i}\otimes\ket{j},
\end{align*}
and then linearly extended to all operators such that
\begin{align*}
  c_1\ketbra{\psi_1}{\phi_1} +
  c_2\ketbra{\psi_2}{\phi_2} \mapsto
    c_1\ket{\psi_1}\otimes\ket{\phi_1^*} 
    + c_2\ket{\psi_2}\otimes\ket{\phi_2^*}.
\end{align*}
Above, $\ket{\phi^*}$ is the state one obtains by complex
conjugating the coefficients of $\ket{\phi}$ in the standard basis.
We note that the vectorization map preserves inner products, $(A, B)
= \Tr(A^\dagger B) = \dbraket{A|B}$. It can also be extended
naturally to many-body setup where we have a lattice $\Lambda$ and
the global space is $\BBH = \bigotimes_{x\in \Lambda} \BBH_x$. In
such case, we view $\BBH\otimes \BBH$ as composite system in which
at every site $x\in \Lambda$, there are \emph{two} copies of the
Hilbert space $\BBH_x$: the original $\BBH_x$ and a fictitious
$\BBH_x$, so that $\mcH\otimes\mcH = \bigotimes_{x\in \Lambda}
(\mcH_x\otimes \mcH_x)$ as described in \autoref{fig:vectorization}.
With this definition, the notion of locality in $\BBH$ can be
naturally mapped to locality in $\BBH\otimes\BBH$. This point of
view is very intuitive when considering the vectorization
$\dket{\sigma^{1/2}}$ of a quantum state $\sigma$. First note
that $\dket{\sigma^{1/2}}$ is a normalized vector since by the
preservation of inner product, 
\begin{align*}
  \dbraket{\sigma^{1/2}|\sigma^{1/2}} 
    = \braket{\sigma^{1/2},\sigma^{1/2}}
    = \Tr(\sigma^{1/2}\cdot \sigma^{1/2}) = \Tr(\sigma)=1.
\end{align*}
Moreover, it is easy to see that 
for every observable $A$ on
$\BBH$, 
\begin{align} \label{eq:EV_sigma}
  \dbra{\sigma^{1/2}}A\otimes\Id \dket{\sigma^{1/2}} 
    = \braket{\sigma^{1/2}, A\sigma^{1/2}} = \Tr(\sigma A),
\end{align}
i.e., the expectation value of every observable $A$ with respect to
$\sigma$ is equal to the expectation value of $A\otimes\Id$ with
respect to $\dket{\sigma^{1/2}}$. This implies that
$\dket{\sigma^{1/2}}$ is a \emph{purification} of
$\sigma$
\begin{align*}
  \sigma = \Tr_{\mathrm{fict}} \dketbra{\sigma^{1/2}}{\sigma^{1/2}},
\end{align*}
where $\Tr_{\mathrm{fict}}$ denotes the tracing over the fictitious
$\BBH_x$ Hilbert spaces. The above observations can be summarized in
the following easy lemma, which relates the entanglement structure
of $\dket{\sigma^{1/2}}$ to that of $\sigma$.
\begin{lemma}
\label{lem:vectorization}
  Let $\sigma$ be a many-body quantum state on a $D$-dimensional
  lattice $\Lambda$, let $\dket{\sigma^{1/2}}$ be the vectorization
  of $\sigma^{1/2}$, and define $\rho\EqDef
  \dketbra{\sigma^{1/2}}{\sigma^{1/2}}$. Then:
  \begin{enumerate}
    \item For any bi-partition $\Lambda = A\cup B$,
      \begin{align}
        I(A:B)_\sigma \le 2S(A)_\rho,
      \end{align}
      where $I(A:B)_\sigma$ is the mutual information between $A,B$
      in the state $\sigma$, and $S(A)_\rho$ is the entanglement
      entropy of region $A$ in the composite system $\BBH\otimes
      \BBH$ (where for every $x\in A$ we include \emph{both} the
      original system and its fictitious partner) with respect to
      the state $\dket{\sigma^{1/2}}$. $S(A)_\rho$ is also known as
      the operator-space entanglement entropy of
      $\sqrt\sigma$ (see
      references\cc{prosen2007operator, jonay2018coarse}).
      
    \item If $\Lambda$ is 1D and there exists an MPS $\ket{\psi_D}$
      on the composite system with bond dimension $D$ such that
      $\norm{\dket{\sigma^{1/2}}-\ket{\psi_D}}\le \epsilon$, then
      $\Psi_{D^2}\EqDef \Tr_{\mathrm{fict}} \ketbra{\psi_D}{\psi_D}$ can
      be described by an MPO with bond dimension $D^2$ and $\norm{\sigma -
      \Psi_{D^2}}_1\le \sqrt 2\epsilon$. A similar relation exists
      also for higher dimensions (replacing MPS with, say, PEPS).
  \end{enumerate}
\end{lemma}
\begin{proof}\ 
  \begin{enumerate}
    \item By definition, $I(A:B)_\rho = S(A)_\rho + S(B)_\rho - 
      S(AB)_\rho$, and as $\rho$ is a pure state and $A,B$ a
      bi-partition of the system, $S(A)_\rho = S(B)_\rho$ and
      $S(AB)_\rho = 0$. Therefore, $I(A:B)_\rho = 2S(A)_\rho$. As
      $\sigma = \Tr_{\mathrm{fict}}\rho$, it follows from the
      monotonicity of the relative entropy that
      \begin{align*}
        I(A:B)_\sigma = S(\sigma||\sigma_A\otimes \sigma_B) \le
          S(\rho||\rho_A\otimes \rho_B) = I(A:B)_\rho = 2S(A)_\rho.
      \end{align*}
      
    \item First note that 
      \begin{align*}
        \epsilon \ge \norm{\dket{\sigma^{1/2}} - \ket{\psi_D}}
          = \sqrt{2(1-\Re\langle\psi_D\dket{\sigma^{1/2}})}
          \ge \sqrt{2(1-|\langle\psi_D\dket{\sigma^{1/2}}|)}
      \end{align*}
      and therefore $|\langle\psi_D\dket{\sigma^{1/2}}| \ge
      1-\epsilon^2/2$ and 
      \begin{align*}
        \norm{\rho - \ketbra{\psi_D}{\psi_D} }_1 =
        \norm{\dketbra{\sigma^{1/2}}{\sigma^{1/2}}
          - \ketbra{\psi_D}{\psi_D} }_1 =
          2 \sqrt{1-|\dbra{\sigma^{1/2}}\psi_D\rangle|^2} 
          \le \sqrt 2 \epsilon.
      \end{align*}
      Therefore, by the monotonicity of the trace distance, 
      \begin{align*}
        \norm{\sigma - \Psi_{D^2} }_1 =
        \|\Tr_{\mathrm{fict}}\rho - \Tr_{\mathrm{fict}}\ketbra{\psi_D}{\psi_D}\|_1
        \le \sqrt 2 \epsilon.
      \end{align*}
      It remains to show that $\Psi_{D^2}$ can be written as an MPO
      with bond dimension $D^2$. This can be understood from
      inspecting \autoref{fig:MPS-MPO}, which shows how $\Psi_{D^2}$
      is obtained by contracting the fictitious legs in the TN
      that describes $\rho=\dketbra{\sigma^{1/2}}{\sigma^{1/2}}$.
  \end{enumerate}
  
\end{proof}

\begin{figure}
  \begin{center}
    \includegraphics[scale=1]{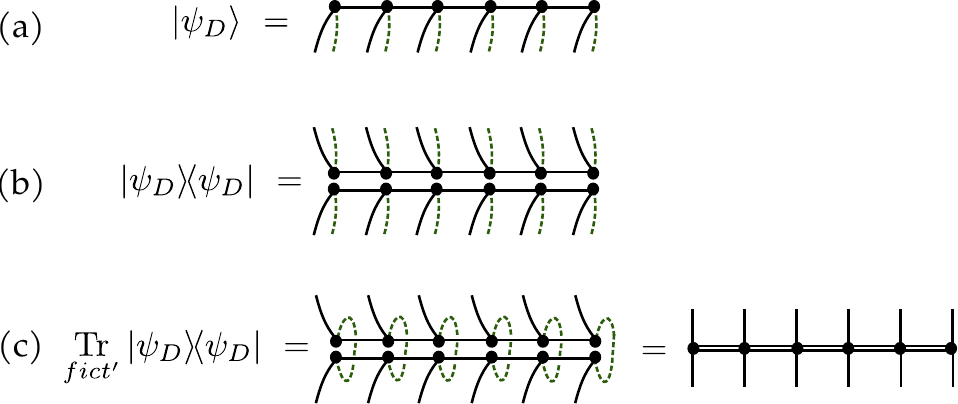}
  \end{center}
  \caption{\label{fig:MPS-MPO} Obtaining an MPO description of
    $\Tr_{\mathrm{fict}}\ketbra{\psi_D}{\psi_D}$ from the MPS
    $\ket{\psi_D}$. (a) The MPS $\ket{\psi_D}$. Each leg is
    represented by two legs, where the solid leg corresponds to the original local 
    Hilbert space and the dashed leg to the fictitious Hilbert
    space. (b) The TN representation of $\ketbra{\psi_D}{\psi_D}$.
    (c) To calculate $\Tr_{\mathrm{fict}}\ketbra{\psi_D}{\psi_D}$ we contract the
    fictitious legs, leading to an MPO with bond dimension $D^2$. }
\end{figure}

We conclude the discussion by noting that vectorization also maps
super-operators to operators. Indeed, for every super operator
$\mcF$ there is a unique $F\in L(\BBH\otimes\BBH)$ that satisfies
\begin{align*}
  \dket{\mcF(X)} \EqDef F\dket{X} .
\end{align*}
This map respects the notion of adjoints in the sense that $\mcF
\mapsto F$ iff $\mcF^*\mapsto F^\dagger$, and so a self-adjoint
super Hamiltonian $\mcH$ is mapped to an hermitian Hamiltonian $H$.
It is easy to verify that if $\mcF$ is the superoperator $\mcF(X)
\EqDef A\cdot X\cdot B$, then its vectorization is the operator $F =
A\otimes B^T$. Using this formula it is clear that a $k$-body
super-operator maps to a $2k$-body operator. For example, the
super Hamiltonian \eqref{eq:super-H-1} from \autoref{thm:local} is
mapped to the Hamiltonian
\begin{align*}
  H = -\sum_j\left[\sqrt{c_j c_{\pi(j)}} L_j\otimes L_j^*
   - \frac{c_j}{2}\big(L_j^\dagger L_j\otimes \Id 
     + \Id \otimes (L_j^\dagger L_j)^T\big)\right] , 
\end{align*}
and a similar expression arises for the super Hamiltonian from
\autoref{thm:quasi-local}.

%--------------------------------------------------------------
\subsubsection{Area laws for steady states}

In this subsection we use the notation $\tilde{O}(X)$ for
$\bigO{X\cdot \poly(\log X)}$. Our first result will be an area law
for 1D Lindbladians that satisfy the requirements of
\autoref{thm:local}. These map to geometrically local 1D super
Hamiltonians, for which we will use the following result:
%%%% theorem - Arad-Kitaev %%%%
\begin{theorem}[Taken from~Refs.\cc{ref:Arad2013-1DAL, 
  ref:Landau2015-1Dalg}]\label{thm:ref:Arad2013-1DAL} 
  Let $\Lambda$ be a 1D lattice of sites with local dimension $d$.
  Let $H=\sum_i h_i$ be a nearest neighbours Hamiltonian with a
  spectral gap $\gamma>0$ and a unique ground state $\ket{\gs}$.
  Assume, in addition, that $\norm{h_i}\le J$ for every $i$, where
  $J$ is some energy scale. Then the entanglement entropy of
  $\ket{\gs}$ with respect to any bi-partition satisfies $S_E
  =\tilde{O}\big( \frac{\log^3(d)}{\gamma/J} \big)$. Moreover, there
  is a matrix product state (MPS) $\ket{\psi_D}$ of sublinear bond
  dimension $D=e^{\tilde{O}(\log^{3/4}n/(\gamma/J)^{1/4})}$ such
  that $\norm{\ket{\gs}-\ket{\psi_D}}\leq \frac{1}{\poly(n)}$, and
  such MPS can be found efficiently on a classical computer.
\end{theorem}

Using the above theorem together with \Lem{lem:vectorization} and
\autoref{thm:local}, we immediately obtain the following corollary
\begin{corol} \label{corol:3.7}
  Let $\mcL$ be a geometrically local $2$-body Lindbladian defined
  on a 1D lattice of local dimension $d$ that satisfies the
  requirements of \autoref{thm:local}. Assume in addition that 
  $\mcL$ has a gap $\gamma>0$, a unique steady state, and that 
  $|c_j|\le J$ for all $j$. Then:
  \begin{enumerate}
    \item $\sqrt{\sigma}$ satisfies an area law for the 
      operator-space entanglement entropy\cc{prosen2007operator},
      that is, for any cut in the 1D lattice into $\Lambda = L \cup
      R$ the entanglement entropy of $\dket{\sqrt\sigma}$ between
      the two parts satisfies
      \begin{align}
        S(L)_\rho = \tilde{O}\big(
          \frac{\log^3(d)}{\gamma/J}\big).
      \end{align} 
    
    \item $\sigma$ satisfies an area-law for the mutual information.
      Specifically, for any cut in the 1D lattice, the mutual
      information between the two sides satisfies 
      \begin{align} 
        \label{eq:sigma_area_law}
        I(L:R)_\sigma = \tilde{O}\big(
          \frac{\log^3(d)}{\gamma/J}\big).
      \end{align}
      
      \item There is a matrix product operator (MPO) 
        $\Psi_D$ with bond dimension
        $D=e^{\tilde{O}(\log^{3/4}n/(\gamma/J)^{1/4})}$ such that
        $\norm{\sigma-\Psi_D}_1\le \frac{1}{\poly(n)}$. Moreover, a
        similar approximating MPO for $\sigma$ can be found
        efficiently on a classical computer.
  \end{enumerate}
\end{corol}

For the next result, we will use the main result of
\cRef{ref:Kuwajara2020-long-range-AL} adapted to the $2$-local
settings with exponentially decaying interactions.

%%%% theorem - Tomotaka %%%%

\begin{theorem}[Taken from~\cRef{ref:Kuwajara2020-long-range-AL}]
\label{thm:Tomotaka} Let $\Lambda$ a 1D lattice of sites with local
  dimension $d$. Let $H=\sum_{i,j} h_{ij}+\sum_i h_i$ be a $2$-body
  Hamiltonian defined on $\Lambda$, and suppose that there exist a
  constant $J$ such that $\norm{h_{ij}} \le \frac{J}{r_{ij}^\alpha},\;
  \norm{h_i}\le J$, and that $H$ has a spectral gap $\gamma>0$ and a
  unique ground state $\ket{\gs}$. Then the entanglement entropy of
  $\ket \gs$ with respect to any cut in the 1D grid satisfies 
  \begin{align*}
    S_E = \tilde{O}\left(\log^2d\left(
      \frac{\log d}{\gamma/J}\right)^{1+ \frac{2}{\alpha-2}}
    \right).
  \end{align*}
  Moreover, there is a MPS $\ket{\psi_D}$ of bond dimension
  $D=e^{\bigO{\log^{5/2}(n)}}$ such that
  $\norm{\ket{\gs}-\ket{\psi_D}}\leq \frac{1}{\poly(n)}$.
\end{theorem}
It should be noted that \autoref{thm:Tomotaka} can be restated for
1D $k$-local Hamiltonians and still give the same entropy bound ---
see Appendix B in \cRef{ref:Kuwajara2020-long-range-AL}. Using the above theorem,
together with \Lem{lem:vectorization} and
Corollary~\ref{cor:quasi-local}, we can choose an
appropriate $B$ for which $\norm{\mcH_{ab}}\le B/|a-b|^\alpha$ for
any $P_a,P_b$ for some $\alpha > 4$. This implies the following
corollary:
\begin{corol} \label{corol:3.9}
  Let $\mcL$ be a geometrically local $2$-body Lindbladian defined
  on a $1D$ lattice that satisfies the requirements of
  \autoref{thm:quasi-local}, with a spectral gap $\gamma>0$ and a steady
  state $\sigma$. Assume also that the smallest non-vanishing
  eigenvalue of $C\cdot C^*$ is $\lambda_{\mathrm{min}}=\bOmega{1}$. Then
  \begin{enumerate}
    \item $\sqrt \sigma$ satisfies an area law for the operator-space
      entanglement entropy\cc{prosen2007operator}.  That is, for any
      cut in the 1D lattice into $\Lambda = L \cup R$ the
      entanglement entropy of $\dket{\sqrt\sigma}$ between the two
      parts satisfies
      \begin{align}
        S(L)_\rho = \tilde{O}\left(\log^2d\left(
          \frac{\log d}{\gamma/J}\right)^{2}\right).
      \end{align} 
    
    \item $\sigma$ satisfies an area-law for the mutual information.
      Specifically, for any cut in the 1D lattice, the mutual
      information between the two sides satisfies 
      \begin{align} 
        \label{eq:sigma_area_law2}
        I(L:R)_\sigma =  \tilde{O}\left(\log^2d\left(
          \frac{\log(d }{\gamma/J}\right)^{2}\right).
      \end{align}
      
      \item There is a matrix product operator (MPO) 
        $\Psi_D$ with bond dimension   $D=e^{\bigO{\log^{5/2}(n)}}$ such that
        $\norm{\sigma-\Psi_D}_1\le \frac{1}{\poly(n)}$.
  \end{enumerate}
\end{corol}

{~}

We finish this part by noting that there are many more applications of the
theorems in \autoref{sec:mapping} which we did not address.
For example:
\begin{enumerate}
    \item One can import other ground
      state area-law results to the Linbladian settings, including
      results for 2D and higher dimensions such as those given in
      Refs.~\cc{ref:Masanes2009-AL, ref:Spyridon2012-adiabaticAL,
      ref:Jaeyoon2014-AL, ref:Brandao2015-sp-heatAL,
      ref:Anshu2021-2DAL}.

    \item One can use the results in Refs.\cc{Has_LSM, 
      hastings2021gapped} to show exponential decay of correlations
      for $\dket{\sqrt\sigma}$, which implies the same for $\sigma$
      using \Eq{eq:EV_sigma}. This was already shown in \cRef{DOC}
      (see Theorem 9 in the paper), so we omit the details. 
      Notice that our result are implied for any such Lindbladian without assuming local uniqueness (see \emph{regular} Lindbladians in \cRef{DOC}).

    \item One can demonstrate a gap in the Lindbladian by demonstrating a gap
      in the super-Hamiltonian instead. As the former task is
      somewhat challenging, the latter has been discussed more
      frequently in the literature.
      Below we will exemplify this by using the
      finite-size criteria from \cRef{knabe} to demonstrate a gap for a family of
      Lindbladians (see the proof of Claim~\ref{clm:local_toy_model},
      bullet~\ref{clm:bul4} for details). We remark that this type of work has been also done in \cRef{alicki2009}.
\end{enumerate}

%================================================================== 
\subsection{Proofs of Theorems~\ref{thm:local}, \ref{thm:quasi-local}}
\label{sec:Proofs}

In this section we give the proofs of our main results,
Theorems~\ref{thm:local}, \ref{thm:quasi-local}. The two proofs
follow the same outline: Examining the continuous family of
superoperators $\mcL_x\EqDef \Gamma_s^{-x}\circ\mcL\circ
\Gamma_s^x$, we would like to show that it obeys some notion of locality  at $x=1/2$. In fact, we will show that it is local for all $x\in
\BBR$. By definition, $\mcL_0 = \mcL$, and by the QDB
condition~\eqref{eq:QDB}, $\mcL_1 = \Gamma_s^{-1}\circ\mcL\circ
\Gamma_s = \mcL^*$, both of which are local. To prove locality
for all other $x\in \BBR$ we use the modular basis (see
Corollary~\ref{corol:modular-basis}). We will show that $\mcL_x$
remains diagonal in the modular basis for \emph{every} $x$. Then, we
will use \autoref{thm:freedom} to establish a connection between the
local representation of $\mcL$ and the modular basis. Finally, we
will employ the fact that $\mcL_1$ is local to argue that it must remain
local for all $x\in\BBR$. %\mrep{x∈\BBRx\in\BBR}{x∈[0,1]x\in[0,1]}.

Let us then begin by proving that $\mcL_x$ remains diagonal in the
modular basis.

\begin{lemma} (Adaptation of Lemma 7 from \cRef{Szegedy})
\label{lem:modular-diag} \ \\ Let $\mcL$ satisfy the quantum
  detailed-balance condition with respect to $s\in
  [0,1/2)\cup(1/2,1]$. Then for every $x\in
  \BBR$, the superoperator $\Gamma_s^{-x}\circ \mcL\circ
  \Gamma_s^x$ is diagonal in the modular basis, and is given by
  \begin{align*}
    (\Gamma_s^{-x} \circ \mcL \circ \Gamma_s^x)(A)
    &= \sum_{\alpha\in I} \gamma_\alpha e^{-\omega_\alpha /2}
      \Big( e^{x\omega_\alpha}S_\alpha A S_\alpha^\dagger
        - \frac{1}{2}\big\{S_\alpha^\dagger S_\alpha, A\big\}\Big).
  \end{align*}
  \textbf{Note:} the above formula is identical to the formula for
  $\mcL$ in \autoref{thm:canonical_form}, except for the
  $e^{x\omega_\alpha}$ factor in front of the $S_\alpha A
  S_\alpha^\dagger$ term.
\end{lemma}

\begin{proof}
  Recall that $\Gamma_s^x(A)=\sigma^{x(1-s)}A\sigma^{xs}$. We
  consider the terms $S_\alpha A S_\alpha ^\dagger$ and
  $\{S_\alpha^\dagger S_\alpha,A \}$ separately. The conjugation of
  $S_\alpha A S_\alpha ^\dagger$ gives
  \begin{align} 
  \label{eq:transformation}
    \big(\Gamma_s^{-x} \circ (A\mapsto S_\alpha A S_\alpha^\dagger) 
      \circ \Gamma_s^x\big)(A) &=
    \Gamma_s^{-x}\big(S_\alpha\Gamma_s^x (A)S_\alpha^\dagger\big)
      = \sigma^{-x(1-s)}S_\alpha \sigma^{x(1-s)}
      A\sigma^{xs} S_\alpha^\dagger\sigma^{-xs}  .
  \end{align}
  Recalling that 
  \begin{align*}
    \sigma^{-x(1-s)} S_\alpha  \sigma^{x(1-s)}
      &= \Delta_\sigma^{-x(1-s)}(S_\alpha) 
        = e^{x(1-s)\omega_\alpha}S_\alpha,\\
    \sigma^{xs}S^\dagger_\alpha \sigma^{-xs} 
      &= \Delta_\sigma^{xs}(S^\dagger_\alpha)
      = e^{xs\omega_\alpha} S^\dagger_\alpha,  \\
  \end{align*}
  we find that, overall, $\big(\Gamma_s^{-x} \circ (A\mapsto S_\alpha A
  S_\alpha^\dagger) \circ \Gamma_s^x\big)(A) =
  e^{x\omega_\alpha}S_\alpha A S_\alpha^\dagger$. 
  
  For the the anti-commutator conjugation, we have:
  \begin{align*}
    \big(\Gamma_s^{-x} \circ (A\mapsto \{S^\dagger_\alpha S_\alpha, A\})
      \circ \Gamma_s^x\big)(A) &= \Gamma_s^{-x}\big( 
       \{S_\alpha^\dagger S_\alpha, \Gamma_s^{x}(A)\}\big) 
     = \sigma^{-x(1-s)} \{S_\alpha^\dagger S_\alpha,
       \sigma^{x(1-s)}A\sigma^{xs}\}\sigma^{-xs}\\
    &= \sigma^{-x(1-s)} S_\alpha^\dagger S_\alpha \sigma^{x(1-s)}A
      + A\sigma^{xs}S_\alpha^\dagger S_\alpha \sigma^{-xs}
    = \{S_\alpha^\dagger S_\alpha, A\} ,
  \end{align*}
  where in the last equality we used the fact that for every $t$,
  \begin{align*}
    \sigma^t S_\alpha^\dagger S_\alpha \sigma^{-t}
     = \sigma^t S_\alpha^\dagger\sigma^{-t}
       \sigma^{t} S_\alpha \sigma^{-t} =
       \Delta_\sigma^{t}(S_\alpha) \Delta_\sigma^{t}(S^\dagger_\alpha)
     = e^{-t\omega_\alpha}S_\alpha e^{t\omega_\alpha}S^\dagger_\alpha,
     = S_\alpha S^\dagger_\alpha .
  \end{align*}
\end{proof}

With \Lem{lem:modular-diag} at hand, we turn to the proof of
\autoref{thm:local}.

%--------------------------------------------------------------------
\subsubsection{Proof of \autoref{thm:local}}
\label{sec:proof-local}

By assumption, $\mcL$ is given by
\begin{align*}
  \mcL(A) = \sum_j c_j \big(L_jA L_j^\dagger 
    - \frac{1}{2}\{ L_j^\dagger L_j, A\}\big),
\end{align*}
but in addition, by the QDB condition and
\autoref{thm:canonical_form}, it is given by a canonical form
\begin{align*}
  \mcL(A) = \sum_{\alpha\in I} \gamma_\alpha e^{-\omega_\alpha /2}
    \Big( S_\alpha A S_\alpha^\dagger 
    - \frac{1}{2}\big\{S_\alpha^\dagger S_\alpha, A\big\}\Big).
\end{align*}
Therefore, by \autoref{thm:freedom}, there must be a unitary
$U_{\alpha,j}$ that connects these two bases, i.e.,
\begin{align}
\label{eq:L-S-connection}
  \sqrt{\gamma_\alpha}e^{-\omega_\alpha/4} S_\alpha = \sum_j
  U_{\alpha, j} \sqrt{c_j}L_j .
\end{align}
Let us now consider the superoperator $\mcL_x \EqDef \Gamma_s^{-x}
\circ \mcL \circ \Gamma_s^x$. By \Lem{lem:modular-diag}, it is
given by
\begin{align*}
  \mcL_x(A) = \sum_{\alpha\in I} \gamma_\alpha e^{-\omega_\alpha /2}
    \Big( e^{x\omega_\alpha}S_\alpha A S_\alpha^\dagger
      - \frac{1}{2}\big\{S_\alpha^\dagger S_\alpha, A\big\}\Big),
\end{align*}
and therefore by \Eq{eq:L-S-connection}, it can also be 
written in terms of the $L_j$ basis as
\begin{align*}
  \mcL_x(A) &= \sum_{\alpha,j,j'} e^{x\omega_\alpha} U_{\alpha, j}
  U_{\alpha, j'}^* \sqrt{c_j c_{j'}} L_j A L_{j'}^\dagger 
    - \frac{1}{2}\sum_j c_j \big\{ L_j^\dagger L_j, A\}.
\end{align*}
Note that the anti-commutator term remained unchanged. To
analyze the first term, we define the matrix $B$ 
\begin{align*}
  B_{j'j}\EqDef\sum_\alpha e^{\omega_\alpha}
      U_{\alpha,j}U_{\alpha, j'}^*.
\end{align*}
Note that $B=U^\dagger\cdot \diag(\{e^{\omega_\alpha}\})\cdot U$,
and therefore $B^x = U^\dagger\cdot\diag(\{e^{x\omega_\alpha}\})
\cdot U$, hence
\begin{align}
\label{eq:Lx}
  \mcL_x(A) = \sum_{j,j'} B^x_{j'j} \sqrt{c_{j'}c_j } L_j A L^\dagger_{j'}
    - \frac{1}{2}\sum_j c_j\big\{ L_j^\dagger L_j, A\} .
\end{align}
To prove the locality of $\mcL_x$, we would like to show that
$B_{j'j}^x \propto \delta_{jj'}$, , i.e., that $B^x$ is diagonal,
for at least one $\tilde{x}\neq 0$, which will then imply that it
also diagonal for \emph{all} $x\in \BBR$. This is done using the
QDB condition~\eqref{eq:QDB}, which will let show diagonality
for $\tilde{x}=1$. The QDB condition is
equivalent to 
\begin{align*}
  \mcL_1(A) = \mcL^* (A) = \sum_j c_j \big(L^\dagger_jA L_j
    - \frac{1}{2}\{ L_j^\dagger L_j, A\}\big).
\end{align*}
Comparing to \Eq{eq:Lx} gives the condition
\begin{align*}
  \sum_{j'j} B_{j'j} \sqrt{c_{j'}c_j } L_j A L^\dagger_{j'}
   = \sum_j c_j L^\dagger_jA L_j 
   = \sum_{j} c_{\pi(j)} L_j A L_j^\dagger, 
\end{align*}
where in the last equality we used the definition of $\pi(j)$,
$L_{\pi(j)} = L^\dagger_j$. Let us now use the assumption that the
jump operators $\{L_j\}$ are linearly independent. This implies that
the coefficients of $L_j A L_{j'}^\dagger$ on both sides of the
equation should be identical, and therefore,
\begin{align*}
  B_{j'j}\sqrt{c_{j'}c_j} = c_{\pi(j)}\delta_{j'j}.
\end{align*}
Using the assumption that $c_j>0$,  we conclude that
\begin{align*}
  B_{j'j} = \frac{c_{\pi(j)}}{c_j} \delta_{j'j},
\end{align*}
and therefore,
\begin{align}
  \mcL_x(A) = \sum_j \Big( c_{\pi(j)}^x c^{1-x}_j L_j A L^\dagger_j
    - \frac{c_j}{2} \big\{ L_j^\dagger L_j, A\}\Big) .
\end{align}
Substituting $x=1/2$ concludes the proof. \qedsymb

%--------------------------------------------------------------------
\subsubsection{Proof of \autoref{thm:quasi-local}}
\label{sec:proof-quasi-local}

As in the proof of \autoref{thm:local}, we show the locality of
$\mcH$ using its modular basis representation. We start by noting
that all the modular basis elements $S_\alpha$ that appear in $\mcL$
in~\eqref{def:canonical_form} (those with $\gamma_\alpha>0$) can be
unitarily expressed by the orthonormal basis $\{P_a\}$. Indeed,
starting from the jump operators representation in
\Eq{eq:Lindbladian2}, we conclude from \autoref{thm:freedom} that
whenever $\gamma_\alpha>0$, $S_\alpha$ can be written in terms of
the jump operators $L_j$. But as $L_j$ can be expanded in terms of
the \emph{local} $\{P_a\}$ operators, it follows that this also holds for
$S_\alpha$ as well. Finally, by the orthonormality of both sets of
operators, we conclude that they are unitarily related: There exists
a unitary $U_{a\alpha}$, where $\alpha$ runs over all the indices
$\alpha$ for which $\gamma_\alpha>0$, such that
\begin{align}
\label{eq:S-P-connection}
  S_\alpha = \sum_a U_{a\alpha} P_a .
\end{align}
Plugging the expansion~\eqref{eq:S-P-connection} into the canonical
form~\eqref{def:canonical_form} gives
\begin{align*}
 \mcL(\rho) = \sum_{a,b,\alpha} U_{a\alpha}\gamma_\alpha  
   e^{-\omega_\alpha/2} U_{b\alpha}^* 
     \big(P_a\rho P_b -\frac{1}{2}\{ P_b P_a,\rho\}\big) .
\end{align*}
Comparing with \Eq{eq:LP}, we use the fact that once an orthonormal
basis is fixed, the coefficients of the Lindbladians are also fixed
(see, e.g., Theorem~2.2 in \cRef{Gorini}). Therefore,
\begin{align*}
     C &= U D U^\dagger, & 
     D &\EqDef \diag \big(\{\gamma_\alpha
     e^{-\omega_\alpha/2}\}\big) ,
\end{align*}
where $C$ is the coefficient matrix of the expansion of $\mcL$ in
terms of the $\{P_a\}$ operators (\Eq{eq:LP}).

We now examine the superoperator $\mcL_x \EqDef
\Gamma_s^{-x}\circ\mcL\circ\Gamma_s^x$, which by
\Lem{lem:modular-diag} is given by
\begin{align*}
   \mcL_x(\rho) = \sum_{\alpha\in I} \gamma_\alpha e^{\omega_\alpha(x-1/2)}
     S_\alpha\rho S_\alpha^\dagger  
     - \frac{1}{2}\sum_{\alpha\in I} \gamma_\alpha  
       e^{-\omega_\alpha/2} \big\{S_\alpha^\dagger S_\alpha, \rho\big\}.
       \end{align*}
Using \Eq{eq:S-P-connection}, we can rewrite it in terms of the
$\{P_a\}$ operators as
\begin{align}
\label{eq:Lx-P}
  \mcL_x(\rho) &=  \sum_{a,b} C_{ab}(x) P_a\rho P_b 
      - \frac{1}{2}\sum_{a,b} C_{ab}\{ P_b P_a,\rho\} ,
\end{align}
where we defined
\begin{align} 
\label{def:Cx}
  C(x) \EqDef U \cdot\diag \left(\{\gamma_\alpha
    e^{\omega_\alpha(x-1/2)}\}\right) \cdot U^\dagger .
\end{align}
Note that $C(x)$ commutes with $C(x^\prime)$ for any $x,x^\prime$,
as they are both diagonalized by $U$. As in the proof of the
previous theorem, we know that $C(x)$ is local for $x=0$ and we
would like to prove locality for every $x\in \BBR$. We do this by
using the QDB condition to show that also the $\tilde{x}=1$ is
local, and by showing that any other $x\in \BBR$, $C(x)$ is a simple
function of $C(0)$ and $C(1)$.

Consider then the $\tilde{x}=1$ point. By the QDB
condition \eqref{eq:QDB}, it follows that
\begin{align*}
  \mcL_1 = \Gamma_s^{-1}\circ\mcL\circ\Gamma_s = \mcL^* 
    = \sum_{a,b} C^*_{ab} \big(P_a \rho P_b
    -\frac{1}{2}\{P_aP_b, \rho\}\big) .
\end{align*}
Comparing it to \Eq{eq:Lx-P} and using the fact that $C$ is
Hermitian, we conclude that $C(1)=C^*$, and therefore, $C^*$
commutes with $C=C(0)$. Finally, a simple algebra shows that
\begin{align*}
  C^{1-x}\cdot (C^*)^x &= U \cdot \diag\big(\{\gamma^{1-x}_\alpha\cdot
    e^{-\omega_\alpha(1-x)/2}\}\big) 
      \cdot\diag\big(\{\gamma^x_\alpha 
      \cdot e^{x \omega_\alpha/2}\}\big) \cdot U^\dagger\\
  &= U \cdot\diag\big(\{\gamma_\alpha e^{\omega_\alpha(x-1/2)}\}\big) 
    \cdot U^\dagger = C(x).
\end{align*}
which can also be written as $C(x) = C^{1-x}(0)\cdot
C^x(1)$. Substituting $x=1/2$, and using \Eq{eq:Lx-P} proves the
theorem.

%%%%%%%%%%%%%%%%%%%%%%%%%%%%%   Examples.  %%%%%%%%%%%%%%%%%%%%%%%%%%%% 
\section{Examples}
\label{sec4}

In this section, we describe two exactly solvable models whose
dynamics is governed by a Lindbladian satisfying, respectively, the
requirements of \autoref{thm:local} and \autoref{thm:quasi-local}
from \autoref{sec:results}. In \autoref{sec4:2} we describe
classical Metropolis-based dynamics that satisfies the conditions of
\autoref{thm:local}, and in \autoref{sec4:1} we describe a family of
quadratic fermion models that satisfy the conditions of
\autoref{thm:quasi-local}. For both models, we prove detailed
balance, unique steady-state and constant spectral gap. As the
models are exactly solvable, we give explicit expressions for their
super Hamiltonians, which are derived independently of the results
of \autoref{sec:mapping}.

%%%%%%%%%   Toy model 1  %%%%%%%%%
\subsection{Model for \autoref{thm:local}: Classical-like Lindbladians}
\label{sec4:2}

Here we describe a system that obeys the requirements of
\autoref{thm:local}.
First, we specify the steady state, which is
the Gibbs state of a classical
Hamiltonian, and then write the corresponding Lindbladian that
annihilates it; this Lindbladian gives rise to classical
thermalization dynamics of the diagonal elements of the density
matrix, while dephasing away the off-diagonal elements. The full
proofs are given in \autoref{App:toy1}.

We consider a 1D lattice with periodic boundary conditions, where a
single qubit occupies each site. Intuitively, we think of each site
$i$ as being able to hold a particle (e.g., a fermion or a hard core
boson) with some energy $\epsilon_i$, or being empty. The state of
the system is then described by a binary string $\ux=(x_1, \ldots,
x_n)\in \{0,1\}^n$, where $x_i$ determines the occupation of the
$i$th site (either $0$ or $1$). To achieve a non-trivial Gibbs
state, we define a constant interaction $u$ acting between nearest
neighbours, which is non-zero if an only if both neighbouring sites
are simultaneously occupied by a particle. To write down the
Hamiltonian of the system, we let
$n_i\EqDef\frac{1}{2}(\Id-Z_i)=\ketbra{1}{1}_i$ denote occupation
number operator at site $i$, and let $\mu>0$ denote a chemical
potential. Our Hamiltonian is then given by
\begin{align*}
  H \EqDef \sum_{i=1}^n (\epsilon_i-\mu)n_i 
    + u\sum_{i=1}^{n} n_i\cdot n_{i+1}. 
\end{align*}
This is a classical Hamiltonian that is diagonal in the
computational basis $\ket{\ux} =
\ket{x_1}\otimes\ldots\otimes\ket{x_n}$, and its Gibbs state is

\begin{align} \label{def:the_steadystate}
  \sigma & = \frac{1}{Z}\sum_{\ux\in \{0,1\}^n}
    e^{-\beta E_x}\ketbra{\ux}{\ux},\\
  E_\ux & \EqDef \sum_{i=1}^n (\epsilon_i-\mu)x_i 
    + u\sum_{i=1}^{n}x_i x_{i+1},\\
  Z & \EqDef \sum_{\ux\in \{0,1\}^n} e^{-\beta E_x}.
\end{align}

We now introduce a Lindbladian that drives the system into $\sigma$.
It is a geometrically $3$-local Lindbladian that is given by
\begin{align}
\label{def:L-toy1}
  \mcL \EqDef \sum_{k,\bi } \gamma_{k,b} \LLkl,
\end{align}
where $k\in\{1,\ldots,n\}$, $\bi\in \{0,1,2\}$ and $\{\gamma_{k,b}\}$ are
some positive (possibly random) $O(1)$ constants. For each $k,\bi$, the
superoperator $\LLkl$ act locally on qubits $k-1,k,k+1$ and
is defined by:
\begin{align}
\label{def:toy1_lindbladian}
  \LLkl(\rho) \EqDef e^{-\beta\omega_{k,b}/2} \left(L_{k,b} \rho L_{k,b}^\dagger 
    - \frac{1}{2}\{L_{k,b}^\dagger L_{k,b} ,\rho\} \right)
  + e^{\beta\omega_{k,b}/2}\left( L_{k,b} ^\dagger \rho L_{k,b}  
    - \frac{1}{2}\left\{L_{k,b}L_{k,b}^\dagger, \rho\right\} \right),
\end{align}
where
$\omega_{k,b} \EqDef -(\epsilon_k-\mu)-u\cdot b$ and %let
$L_{k,b}$ are %be
$3$-local jump operators given by
\begin{align}
\label{def:toy1_jump_operators}
  L_{k,b} \EqDef \sigma^-_k\otimes\Pi^b_{k-1,k+1} 
    \otimes \Id_\mathrm{rest} .
\end{align}
Above, $\sigma^-_k \EqDef \ketbra{0}{1}_k$ is the ``annihilation
operator'' on site $k$, and $\Pi^b_{k-1,k+1}$ is the projector onto
the subspace in which the sum of the occupation numbers of sites
$k-1$ and $k+1$ is equal to $b$ (for $b\in\{0,1,2\}$).

We will now show that $\sigma$ is the unique steady state of $\mcL$,
and that $\mcL$ is gapped and satisfies QDB.  The key is to show
explicitly that the representation of $\mcL$ in
\Eq{def:toy1_lindbladian} is the canonical representation, from
which it will follow by \autoref{thm:canonical_form} that $\mcL$
satisfies QDB with respect to $\sigma$, and therefore it is a steady
state of $\mcL$. Formally, we claim:
\begin{claim} 
\label{clm:local_toy_model} 
  The following properties hold for the Lindbladian $\mcL$ defined
  in \eqref{def:L-toy1}:
  \begin{enumerate}
    \item \label{clm:bul1} The jump operators $\{L_{k,b}, 
      L_{k,b}^\dagger\}$ are taken from a modular basis (up to
      normalization), and so $\mcL$ is given in the canoniconical
      representation.
  
    \item \label{clm:bul2} $\mcL_{k,b}$ satisfy local 
      detailed-balance with respect to $\sigma$,
      \begin{align*}
        \LLkl  \circ \Gamma_s = \Gamma_s \circ \LLkl^*
        \tab \forall s\in [0,1] .
      \end{align*}
  
    \item \label{clm:bul3} If $\gkl>0$, then $\mcL$ has a unique 
      steady state $\sigma$ defined in \Eq{def:the_steadystate}.
    
    \item \label{clm:bul4} If $\gkl\geq \alpha > 0$ and $\epsilon_k 
      - \mu,u = \bigO{1}$, then $\mcL$ is gapped with
      $\gap{\mcL}=\bOmega{\alpha}$.
  \end{enumerate}
\end{claim}

\begin{proof}
  To show that $\mcL$ is given in the canonical form, we will show
  that $\{L_{k,b}\}$ is a proper modular basis with well-defined
  Bohr frequencies given by $\beta\omega_{k b}$. First, it is easy to see that $\{L_{k,b}\}$ are
  orthogonal and traceless.  To see that they have
  well-defined Bohr frequencies, we need to show that
  $\Delta_\sigma(L_{k,b}) = e^{-\omega_{k,b}}L_{k,b}$, where
  $\Delta_\sigma(L_{k,b}) = \sigma L_{k,b} \sigma^{-1}$. By
  \Eq{def:the_steadystate},
  \begin{align}
  \label{eq:Delta-L}
    \Delta_\sigma(L_{k,b}) = \sum_{\ux^\prime,\ux} 
      e^{-\beta(E_{\ux^\prime}-E_{\ux})}
        \ketbra{\ux^\prime}{\ux^\prime}L_{k,b}\ketbra{\ux}{\ux}.
  \end{align}
  However, by the definition of $L_{k,b}$,
  \begin{align}
  \label{eq:Lkb-cross}
    \bra{\ux^\prime}L_{k,b}\ket{\ux} = 
    \begin{cases}
      1 , & x_k=1, \; x^\prime_k=0,\;  x_i=x^\prime_i \; \mathrm{for all} \; i\neq k ,
        \; \mathrm{and}\; x_{k-1}+x_{k+1} = b  \\
      0, &\mathrm{otherwise} .
    \end{cases}
  \end{align}
  Therefore, whenever $\bra{\ux^\prime}L_{k,b}\ket{\ux}\neq 0$, 
  \begin{align*} %\label{eq:Bohr_freq_toy1}
    E_{x^\prime}-E_x = \sum_{i=1}^n (\epsilon_i-\mu)({x^\prime}_i -x_i) 
      + u\sum_{i=1}^{n}({x^\prime}_i{x^\prime}_{i+1}-x_ix_{i+1})
      = -(\epsilon_k-\mu) - ub = \omega_{k,b}.
  \end{align*}
  Plugging this into \Eq{eq:Delta-L}, we get
  \begin{align*}
    \Delta_\sigma(L_{k,b}) = e^{-\beta\omega_{k,b}}\sum_{\ux^\prime,\ux}
      \ketbra{\ux^\prime}{\ux^\prime}L_{k,b}\ketbra{\ux}{\ux} 
      = e^{-\beta\omega_{k,b}}L_{k,b}.
  \end{align*}
  A similar calculation shows that $\Delta_\sigma(L^\dagger_{k,b})
  =e^{\beta\omega_{k,b}}L^\dagger_{k,b}$.

  Using \autoref{thm:canonical_form} it follows that for every
  $k,b$, $\mcL_{k,b}$ satisfies the QDB condition for $s\in
  [0,1/2)\cup(1/2,1]$ (and therefore also for $s=1/2$) with respect
  to $\sigma$, i.e., $\Gamma_s\circ \mcL_{k,b} = \mcL_{k,b}^*\circ
  \Gamma_s$. As shown in the end of \autoref{sec:QDB}, this implies
  that $\sigma$ is annihilated by $\mcL_{k,b}$. Therefore, $\sigma$
  is a fixed point of $\mcL$, and, moreover, it is a
  ``frustration-free'' Lindbladian. 
  
  Showing uniqueness and gap (bullets \ref{clm:bul3},\ref{clm:bul4})
  are technically more involved, requiring the use of Knabe and
  Perron-Frobenious Theorems, and are therefore deferred to
  \autoref{App:toy1}.
  
\end{proof}

Using the claim above, specifically bullet \ref{clm:bul2}, the
superoperator $\mcH_{k,\bi}\EqDef \Gamma_s ^{-1/2}\circ\LLkl \circ
\Gamma_s^{1/2} $ is self-adjoint and annahilates $\sqrt \sigma$.
Then, using bullet~\ref{clm:bul1}, we can apply
\Lem{lem:modular-diag} to $\LLkl$ given in \Eq{def:toy1_lindbladian}
to get
\begin{align}
\label{def:the_superhamiltonian}
\mcH_{k,\bi}  = -\left(
    L_{k,b} \rho L_{k,b}^\dagger + L_{k,b}^\dagger \rho L_{k,b}        
    -\frac{e^{-\beta\omega_{k,b}/2}}{2}\{L_{k,\bi}^\dagger L_{k,b} ,\rho\} 
   -\frac{e^{\beta\omega_{k,b}/2}}{2}\left\{ L_{k,b} L_{k,b}^\dagger , \rho \right\} .
\right)
\end{align}
Consequently, $\mcH=\sum_{k,\bi}\gamma_{k,\bi}\mcH_{k,\bi}$ is a
frustration-free local super-Hamiltonian. While the locality of $\mcH$
is guaranteed by \autoref{thm:local}, frustration-freeness is an
extra feature that follows from that fact that every $\mcL_{k,b}$ is
locally QDB. This allows us to use tools of frustration-free
Hamiltonians to study $\mcH$ and $\mcL$. For example, we prove
bullet \ref{clm:bul4} of Claim~\ref{clm:local_toy_model} using
techniques from Refs.~\cc{knabe,Lemm,jauslin2021random}. We believe
this method may be beneficial for physicists when investigating
detailed-balance semigroups and specifically Davies
generators\cc{davies,kastoryano2016quantum,bardet2021entropy}. As a
final remark, we mention that this could be extended to more
complicated Lindblad operators, higher dimensional lattices, and other degrees of freedom (e.g.,
bosons with a larger finite set of allowed occupancies per site or qudits).

%====================================================================

%%%%%%%%%   Toy model 2   %%%%%%%%%
\subsection{Model for Theorem \ref{thm:quasi-local}: Quadratic fermionic Lindbladians}
\label{sec4:1}

We now describe a family of models that satisfies the requirements
of \autoref{thm:quasi-local}. The models consist of free spinless
fermions on a lattice, coupled to two particle reservoirs, one with
infinitely high chemical potential and one with an infinitely low
chemical potential. The high chemical potential reservoir emits
particles into the system, and the low chemical potential reservoir
takes particles out of the system. The reservoirs are assumed to be
large and Markovian, such that they can be integrated out and result
in a dissipative Lindblad evolution of the lattice. A physical
realization of such settings can be found in, e.g., Section 3 of
\cRef{Moshe_fermions}; such Lindbladians can be used for the
dissipative generation of topologically nontrivial states
\cc{Shavit_Goldstein,Beck_Goldstein}.

The models are governed by quadratic
Lindbladians of the
form\cc{Prosen_2008,Moshe_fermions,barthel2021solving}
\begin{align} \label{def:lindbladian_toy2}
\mcL (\rho)= \sum_{i,j} 
\gin_{i j} \left(a_i^\dagger 
\rho a_j -\frac{1}{2}\{a_j a_i^\dagger,\rho\}\right)
+\gout_{j i}\left( a_i \rho a_j^\dagger
-\frac{1}{2}\{a_j^\dagger a_i,\rho\} \right),
\end{align}
$a_i$, $a_i^\dagger$ are, respectively fermionic annihilation and
creation operators obeying the standard anticommutation relations,
and where $\gin,\gout$ are positive semi definite matrices with
eigenvalues $d_k^\mathrm{in},d_k^\mathrm{out}$, respectively.
Physically, $\gin$ ($\gout$) is responsible for the absorption
(emission) of particles from (into) the environment. This model can
be seen as a special case of a scenario which is treated by\footnote{This 
can be seen, for example, by expanding the fermionic
operators in a Majorana basis\cc{Prosen_2008}, which is Hermitian,
orthonormal and local (in the fermionic sense). Then, expressing the coefficients matrix $C$ in terms of
the elements
of $\gin,\gout$ %, and computing √C⋅C∗\sqrt{C\cdot C^*}. Doing this, we
it is easy to show
that $C\cdot C^*$ is gapped if and only if $\gin$ and
$\gout$ are, and \autoref{thm:quasi-local} will indicate that the
super-Hamiltonian has exponentially decaying interactions.}
%Our proof here, however, avoids this computation and shows this explicitly by exploiting the structure of the Lindbladian written in the fermion annihilation and creation operators.}
\autoref{thm:quasi-local}. However, we find it more illuminating to not use it, but rather prove directly the properties of %independent claims regarding
this specific model.
%The required
Under the assumptions pronounced in the claim below, we show that the Lindbldian satisfies QDB with respect to a unique steady state which is Gaussian. This is proven similarly to \ref{clm:local_toy_model}, by employing the canonical representation of $\mcL$.

%%% First claim %%%
\begin{claim} \label{clm:q-local_toy_model}
  Let $\mcL$ be the Lindbladian in \eqref{def:lindbladian_toy2}.
  Suppose that $\gin$ and $\gout$ are commuting, full-rank matrices. Then
  \begin{enumerate}
  
    \item \label{bul1:toy2} There is a unique steady state of a
      Gaussian form:
        \begin{align*}
        \sigma & = \frac{1}{Z}\exp(-H_{ss}),\\
        \mathrm{where } & \enspace H_{ss}=\sum_{i,j}h_{i j}a_i^\dagger a_j.
        \end{align*}
  
    \item \label{bul2:toy2} $\mcL$ satisfies quantum detailed balance with respect to $\sigma$ for any $s\in[0,1]$.
  
    \item \label{bul3:toy2} Let $\alpha>0$ be a number such that $\gin+\gout\geq \alpha \Id$, then $\gap(\mcL)\geq \frac \alpha 2$.
      In particular, $\mcL$ is gapped if the minimal eigenvalue of $\gin+\gout$ is $\bOmega 1$.
  
  \end{enumerate}
\end{claim}

\begin{proof}
We prove the last statement by transforming into a more convenient single-particle basis for which the Lindbladian decomposes to a sum of single mode sub-Lindbladians. To achieve this, we use the commutativity of $\gin$ and $\gout$, which implies that they can be simultaneously diagonalized:
\begin{align*}
    \gin_{ij}=\sum_{i,j} u_{ik}d^\mathrm{in}_k u_{jk}^*,\\
    \gout_{ij}=\sum_{i,j} u_{ik}d^\mathrm{out}_k u_{jk}^*.
\end{align*}
Plugging this expansion into $\mcL$ gives
\begin{align}\label{eq:canonical_toy_fermions}
\begin{split}
\mcL (\rho) & = \sum_{i,j,k}
u_{ik}d^\mathrm{in}_k u_{jk}^* \left(a_i^\dagger 
\rho a_j -\frac{1}{2}\{ a_ja_i^\dagger,\rho\}\right)
+u_{jk}d^\mathrm{out}_k u_{ik}^*\left( a_i \rho a_j^\dagger
-\frac{1}{2}\{a_j^\dagger a_i,\rho\} \right)\\
& = \sum_{k}
d^\mathrm{in}_k  \left(c_k^\dagger 
\rho c_k -\frac{1}{2}\{ c_k c_k^\dagger,\rho\}\right)
+d^\mathrm{out}_k \left( c_k \rho c_k^\dagger
-\frac{1}{2}\{c_k^\dagger c_k,\rho\} \right)
\EqDef \sum_k \mcL_k(\rho),
\end{split}\end{align}
where we defined $c_k\EqDef \sum_i u_{i,k}^* a_i$, the annihilation operators corresponding to the single-particle eigenmodes of $\gin$ and $\gout$.
Notice that in the new basis, $\mcL$ decomposes into a sum of single mode Lindbladians $\mcL_k$, as promised. 
By calculating the spectrum of each $\mcL_k$, we see that the 
\begin{align*}
    \gap{(\mcL)}
    =\min_k\gap{(\mcL_k)}
    =\min_k{\frac{d^\mathrm{in}_k+d^\mathrm{out}_k}{2}} 
\end{align*}
which is bullet~\ref{bul3:toy2}.
Moreover, each $\mcL_k$ has a unique zero state $\sigma_k$.
Solving for each $\mcL_k$ it is easy to see that it has the form $\sigma_k \propto \exp{\big(-\epsilon_k c_k^\dagger c_k\big)}$ with $\epsilon_k\EqDef -\log\big(\frac{d_k^\mathrm{in}}{d_k^\mathrm{out}}\big)$. The global steady state is then unique, being the product of all single mode zero states, $\sigma=\bigotimes_k \sigma_k \propto \exp{\big(-\sum_k \epsilon_k c_k^\dagger c_k\big)}$. Going back to the original $a_j$ basis completes the proof of bullet~\ref{bul1:toy2}.
%To prove bullets~\ref{bul1:toy2} and~\ref{bul2:toy2}, we assume that the steady state is of the form $\sigma\propto \exp{\big(-\sum_k \epsilon_k c_k^\dagger c_k\big)}$, and tune $\epsilon_k$ so that the Lindbladian in \Eq{eq:canonical_toy_fermions} is given in the canonical form with respect to $\sigma$.
%This will show detailed balance for any $s$ and assure that $\sigma$ is indeed the steady state (see \autoref{thm:canonical_form}).
%To begin with,
Finally, bullet~\ref{bul2:toy2} follows from the observation that $\{c_k,c_k^\dagger \}$ are canonical jump operators, namely, they have well-defined Bohr frequencies:
\begin{align*}
    \sigma c_k \sigma^{-1} = e^{\epsilon_k}c_k ,\\
    \sigma c_k^\dagger \sigma^{-1} = e^{-\epsilon_k}c_k^\dagger .
\end{align*}
The above identities can be easily checked in the eigenbasis of the occupation number operators, $c_k^\dagger c_k$.
%states $\ket{n}$ of $\{c_k\}$.
%Now, to show that \Eq{eq:canonical_toy_fermions} is indeed a canonical form, we need $\epsilon_k$ to satisfy
%\begin{align} \label{eq:cond_DB}
%\frac{d_k^\mathrm{in}}{d_k^\mathrm{out}}=e^{-\epsilon_k} \tab \forall k\in\{1,\ldots,n\}.
%\end{align}
%We choose $\epsilon_k$ that satisfy the equations above, i.e. $\epsilon_k\EqDef -\log\big(\frac{d_k^\mathrm{in}}{d_k^\mathrm{out}}\big)$.
%The form for the steady state in bullet \ref{bul1:toy2} can be achieved by changing back to the $a_j$ basis.
\end{proof}
We remark that: (1) The above can be also proven using correlation matrix formalism and the \emph{continuous Lyaponouv equation}\cc{barthel2021solving}. (2) Starting from a Gaussian steady state and going to the corresponding single-particle eigenbasis one can show that $\gin$,$\gout$ must be diagonal in the same basis. Hence, QDB is satisfied if and only if $\gin$, $\gout$ commute.

As a corollary, the super-Hamiltonian can be derived, thus verifying the results of \autoref{thm:quasi-local} for the current model:
%%% Second claim %%%
\begin{claim} \label{clm:super-Hamiltonian-2}
Let $\mcL$ be a Lindbladian as in \eqref{def:lindbladian_toy2} with
commuting $\gin,\gout$. Then $\mcH=-\glg$ is given by
\begin{align} \label{eq:Hamiltonian_toy1}
    \mcH(\rho) = - \sum_{i,j}  (\sqrt{\gin\cdot \gout})_{i j}a_i^\dagger  \rho   a_j +
    (\sqrt{\gin\cdot \gout})^*_{i j}a_i  \rho a_j^\dagger -  \frac{\gin_{ij}}{2}\{a_j a_i^\dagger ,\rho\} - \frac{\gout_{ji}}{2}\{a_j^\dagger a_i ,\rho\}.
\end{align}
\end{claim}

\begin{proof}
We prove this in a similar fashion to \Lem{lem:modular-diag}.
Recall that $c_k,c_k^\dagger$ are canonical jump operators, that is, $\sigma c_k \sigma^{-1} = e^{\epsilon_k}c_k=\frac{d_k^\mathrm{out}}{d_k^\mathrm{in}}c_k$
and
$\sigma c_k^\dagger \sigma^{-1} =
e^{-\epsilon_k} c_k^\dagger
= \frac{d_k^\mathrm{in}}{d_k^\mathrm{out}}c_k^\dagger$.
Using the Lindblad form in $\Eq{eq:canonical_toy_fermions}$, we see that 
\begin{align*}
    \mcH(\rho) &  = -\glg (\rho) \\
    & = -\sum_k d_k ^{in} (\sigma^{-(1-s)/2} c_k^\dagger \sigma^{(1-s)/2}) \rho (\sigma^{s/2} c_k \sigma^{-s/2})
    - \sum_k d_k ^{out} (\sigma^{-(1-s)/2} c_k \sigma^{(1-s)/2}) \rho (\sigma^{s/2} c_k^\dagger \sigma^{-s/2}) - \frac{1}{2}\{\ldots\} \\
    & = -\sum_k d_k ^{in} (e^{(1-s)\epsilon_k/2} c_k^\dagger )  \rho 
    ( e^{s\epsilon_k/2} c_k  
    - \sum_k d_k ^{out} (e^{-(1-s)\epsilon_k/2} c_k)  \rho (e^{-s\epsilon_k/2} c_k^\dagger ) - \frac{1}{2}\{\ldots\} \\
    & = -\sum_k d_k ^{in} e^{\epsilon_k/2} c_k^\dagger  \rho  c_k
    - \sum_k d_k ^{out} e^{-\epsilon_k/2} c_k \rho c_k^\dagger - \frac{1}{2}\{\ldots\} \\
    & = -\sum_k \sqrt{d_k ^{out}d_k ^{in}}(c_k^\dagger  \rho   c_k  + c_k  \rho c_k^\dagger) - \frac{1}{2}\{\ldots\}
\end{align*}
Similarly to \Lem{lem:modular-diag}, it can be checked that the anti-commutator terms  $\{\ldots\}$ are unchanged by the transformation.
Switching back to the original (local) creation and annihilation operators $a_k$ yields
\begin{align*}
    \mcH(\rho) & =
    -\sum_{i,j,k} u_{i k}\sqrt{d_k ^{out}d_k ^{in}} u_{jk}^* (a_i^\dagger  \rho   a_j)
    + u_{i k}^*\sqrt{d_k ^{out}d_k ^{in}} u_{jk} (a_i  \rho a_j^\dagger) - \frac{1}{2}\{\ldots\} \\
    & = -\sum_{i,j} (\sqrt{\gin\cdot \gout})_{i j}a_i^\dagger  \rho   a_j 
    +(\sqrt{\gin\cdot \gout})_{i j}^* a_i  \rho a_j^\dagger) - \frac{1}{2}\{\ldots\}.
\end{align*}
\end{proof}

Since we are interested in local Lindbladians, we consider fermions
that occupy the sites of a lattice. The locality properties of $\sqrt{\gin\cdot \gout}$ then from \Lem{lem:exp-decay-matrix}.
For the purpose of illustration, let us treat explicitly the case of a 1D lattice
%For simplicity assume it is a 1D lattice
with periodic boundary conditions. We focus on commuting
$\gin,\gout$ that connect nearest neighbours sites only, i.e.
tridiagonal matrices (augmented by the upper-right and lower-left elemetns). Moreover, we consider
translation invariant systems (Toeplitz matrices), such that
\begin{align} \label{eq:local_gammas}
\gin =
\begin{pmatrix}
\gin_0 & \gin_1 & 0 & & \hdots & 0 & \gin_1 \\
\gin_1 & \gin_0 & \gin_1 & & & & 0 \\
0 & \gin_1  &\gin_0   & \\
\vdots & &  &  \ddots &  & & \vdots \\
& & &  & \gin_0 & \gin_1 & 0\\
0 & & & &\gin_1 & \gin_0 & \gin_1 \\
\gin_1 & 0 & & \hdots  & 0 & \gin_1 & \gin_0 
\end{pmatrix}, \quad
\gout = \begin{pmatrix}
\gout_0 & \gout_1 & 0 & & \hdots & 0 & \gout_1 \\
\gout_1 & \gout_0 & \gout_1 & & & & 0 \\
0 & \gout_1  &\gout_0   & \\
\vdots & &  &  \ddots &  & & \vdots \\
& & &  & \gout_0 & \gout_1 & 0\\
0 & & & &\gout_1 & \gout_0 & \gout_1 \\
\gout_1 & 0 & & \hdots  & 0 & \gout_1 & \gout_0 
\end{pmatrix}.
\end{align}

\begin{figure} 
    \begin{subfigure}{.45\textwidth}
        \begin{center}
        \includegraphics[width=1\linewidth]{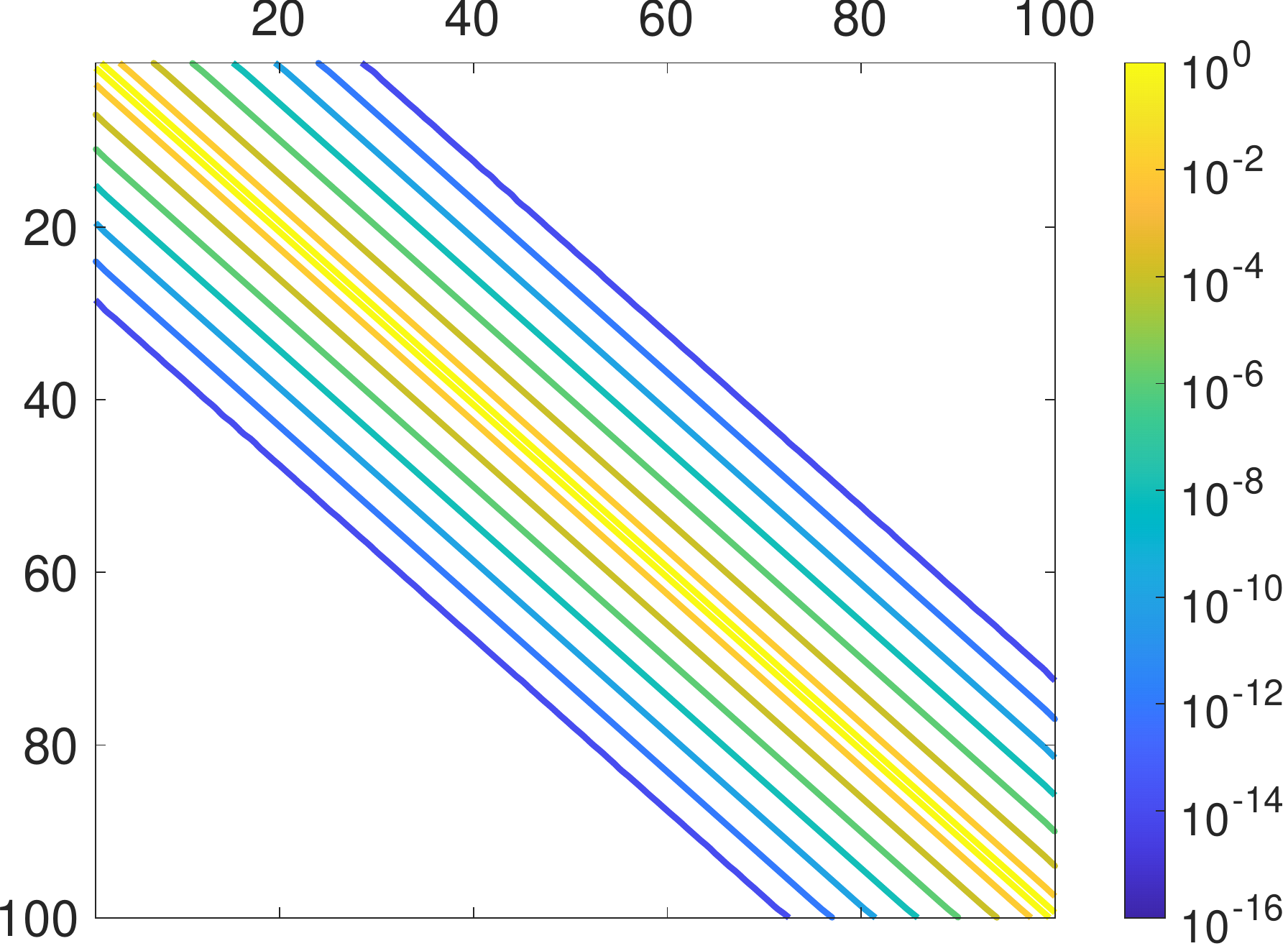} 
        %\caption{Gapped \gin⋅\gout\gin\cdot\gout.}
        \label{fig:graph1}
        \end{center} 
    \end{subfigure}
    \hfill
    %     \begin{subfigure}[h]{0.45\linewidth}
    %     \begin{center}
    %     \includegraphics[width=1\linewidth]{figures/contour31205.pdf} 
    %     \caption{\gin0,\gin1=(3,1)\gin_0,\gin_1=(3,1), \gout0,\gout1=(2,0.5)\gout_0,\gout_1=(2,0.5)}
    %     \label{fig:graph2.2}
    %     \end{center} 
    % \end{subfigure}
    % \vfill
    % \vspace{0.2 cm}
    % \begin{subfigure}[h]{0.45\linewidth}
    %     \begin{center}
    %     \includegraphics[width=1\linewidth]{figures/contour312099.pdf} 
    %     \caption{\gin0,\gin1=(3,1)\gin_0,\gin_1=(3,1), \gout0,\gout1=(2,0.99)\gout_0,\gout_1=(2,0.99)}
    %     \label{fig:graph3.2}
    %     \end{center}
    % \end{subfigure}    
    % \hfill
    \vspace{0.2 cm}
    \begin{subfigure}[h]{0.45\linewidth}
        \begin{center}
        \includegraphics[width=1\linewidth]{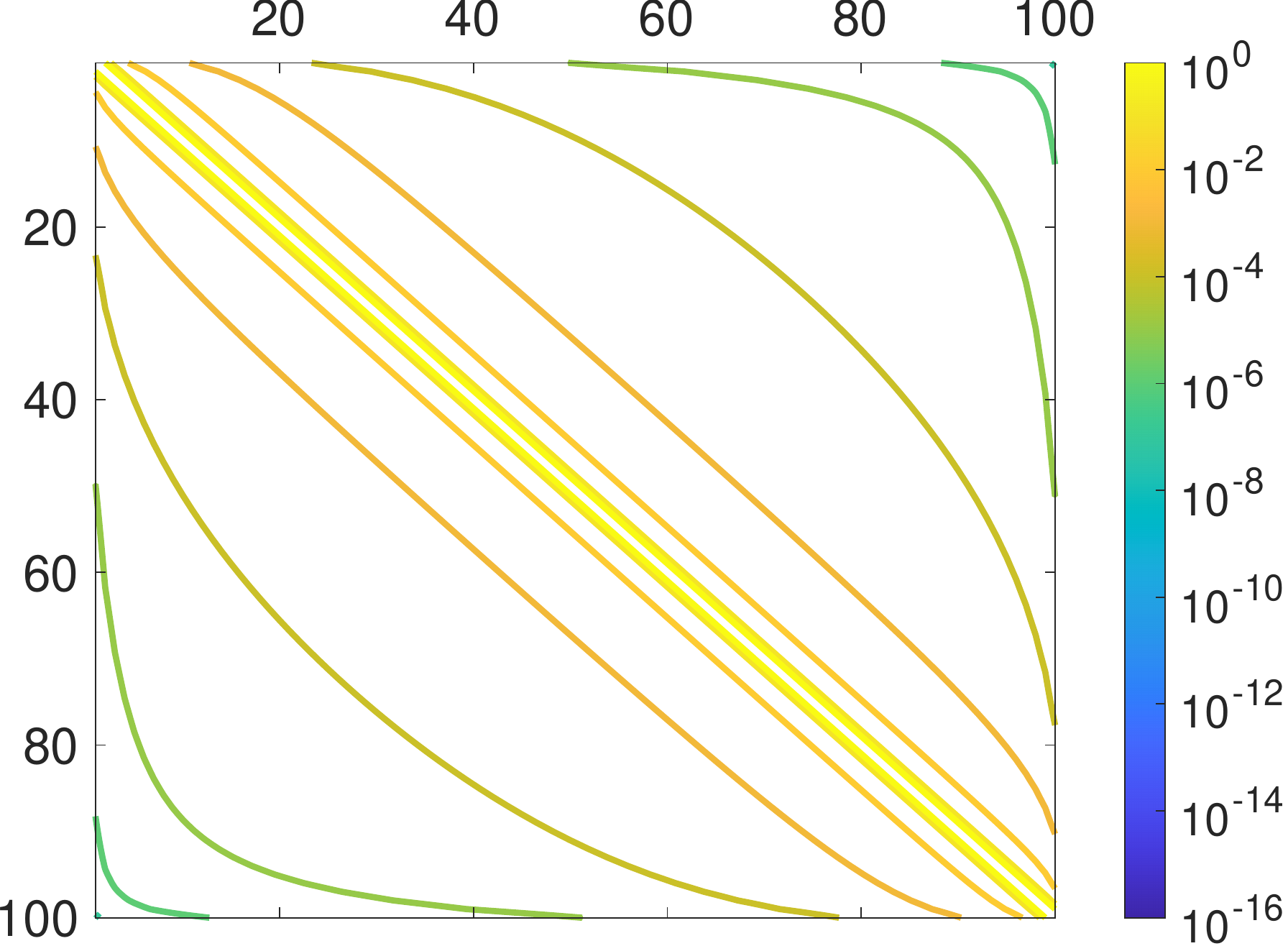} 
        %\caption{Gapless \gin⋅\gout\gin\cdot\gout.}
        \label{fig:graph2}
        \end{center}
    \end{subfigure}
\caption{Plot of the matrix elements,
$(\sqrt{\gin\cdot\gout})_{ij}$, for $n=100$. The horizontal and
vertical axis represent $i$ and $j$, respectively. Left panel:
gapped $\gin\cdot\gout$, with $\gin_0,\gin_1=(3,1)$,
$\gout_0,\gout_1=(2,0.5)$, Right panel: gapless $\gin\cdot\gout$,
with $\gin_0,\gin_1=(3,1)$, $\gout_0,\gout_1=(2,1)$.}
\label{fig:graphs} 
\end{figure}

The canonical jump operators $c_k$ of such a Lindbladian are
achieved by diagonalizing $\gin$ and $\gout$, which, due to
translation invariance, are given by the a discrete Fourier
transform of the original $a_j$:
\begin{align*}
    c_k = \frac{1}{\sqrt{n}}\sum_j e^{-\frac{2\pi i}{n}k  j}a_j, \tab k = 1,\ldots, n.
\end{align*} 
These are in-fact non-local operators, but linear combinations of
such (as anticipated by the proof of \autoref{thm:quasi-local}). By
applying the discrete Fourier transform to \ref{eq:local_gammas}, we
see that $d^\mathrm{in}_k = \gin_0 + 2\gin_1 \cos(\frac{2\pi k}{n})$
and similarly $d^\mathrm{out}_k=\gout_0 + 2\gout_1 \cos(\frac{2\pi
k}{n})$. Therefore, $\gin,\gout$ are gapped\footnote{In the sense
that their smallest non-zero eigenvalue does not vanish as
$n\rightarrow \infty$.} when $\gin_0-2\gin_1=\bOmega{1}>0$ and
$\gout_0-2\gout_1=\bOmega{1}>0$, respectively. We also require that
$\gin_0,\gout_0=\bigO{1}$ for the Lindbladian to have bounded norm.
As we now show, the gaps in $\gin$ and $\gout$ stated above are
responsible for the decay of the matrix elements of
$\sqrt{\gin\cdot\gout}$, and therefore determine the locality of the
super-Hamiltonian \eqref{eq:Hamiltonian_toy1}. This is due to the
following lemma:
%%% Thrid claim %%%
\begin{lemma} \label{lem:gap_in_gamma}
The matrices in \Eq{eq:local_gammas} satisfy
\begin{align*}
    |\sqrt{\gin\cdot \gout})_{i j}|\leq 2\frac{\sqrt{\gin_0}}{\left(1-\frac{2\gin_1}{\gin_0}\right)^2}
    \frac{\sqrt{\gout_0}}{\left(1-\frac{2\gout_1}{\gout_0}\right)^2} \left[
    \left(\frac{2\gin_1}{\gin_0}\right)^{d(i,j)/2}
     +
     \left(\frac{2\gout_1}{\gout_0}\right)^{d(i,j)/2}\right] 
\end{align*}
where $d(i,j)\EqDef \min\{|i-j|,n-|i-j|\}$ is the metric on the circle.
\end{lemma}
As a result, provided that $\gin$ and $\gout$ are gapped, $\mcH$
given in \eqref{eq:Hamiltonian_toy1} is geometrically 2-local
(quadratic) with exponentially decaying interactions. The proof of
\Lem{lem:gap_in_gamma} is technical and thus left to
\autoref{App:toy2}.

We remark that the derivation in \autoref{App:toy2} suggests that
when the gap in $\gin$ (or $\gout$) closes, that is, when
$\gin_0=2\gin_1$, the super-Hamiltonian becomes long range with
polynomially decaying interactions. %, as explained in the appendix.
The degree of the polynomial does not allow an area-law statement
for the steady state, according to the results of
\cRef{ref:Kuwajara2020-long-range-AL}. See also
figure~\ref{fig:graphs} for an illustration of the decay in
$\sqrt{\gin\cdot\gout}$.

%%%%%%%%%%%%%%%%%%%%%%%%%%%%%%%%%%%%%%%%%%%%%%%%%%%%%%%%%%%%%%%%%%%%%
\section{Discussion and further research}

In this work we have shown how a detailed-balance Lindbladian $\mcL$
can be mapped to a local, self-adjoint superoperator $\mcH$, which
we call a super Hamiltonian. The mapping is via a similarity
transformation, hence we are guaranteed that the Lindbladian
and the super Hamiltonian share the same spectrum (up to an overall
minus sign). Moreover, if $\sigma$ is the steady state of the
Lindbladian, $\sqrt{\sigma}$ is the steady state of $\mcH$. By
vectorizing the super Hamiltonian we get a local Hamiltonian whose
ground state is $\dket{\sigma^{1/2}}$. As a side consequence of our
mapping, we also found a necessary condition for a Lindbladian to
satisfy detailed balance, which can be checked efficiently.

We observed that local expectation values in $\sigma$ map to local
expectation values in the ground state $\dket{\sigma^{1/2}}$, and that the
the mutual information in $\sigma$ is bounded by the entanglement entropy in
$\dket{\sigma^{1/2}}$. Consequently, several well-known results
about the structure of gapped ground states of local Hamiltonians
can be imported to the steady state of gapped, detailed-balanced
Lindbladians. In particular, we have shown how under mild conditions
that can be checked efficiently, the steady state of 1D, gapped,
detailed-balanced Lindbladians satisfies an area-law in mutual
information, and can be well approximated by an efficient MPO. These
results cover many new systems for which the results of
Refs.~\cc{Brandao,DOC} are not known to apply.

The mapping applies for Lindbladians with traceless jump operators
and vanishing Hamiltonian part (a consequence of detailed balance).
However, it also applies to Lindbladians with a Hamiltonian that
commutes with the steady state, since it leaves the steady state
invariant.  An example to such a Lindbldian is the Davies
generator\cc{davies} that describes thermalization. The addition of
the corresponding Hamiltonian should not break primitivity\cc{sanz2010quantum},
and we expect that in many cases it will not close the spectral gap.

Our work leaves several open questions and possible directions for
future research. First, it would be interesting to see what other
results/techniques can be imported from local Hamiltonians to
Lindbladians using our mapping. It would also be interesting to see
if this mapping can be used in numerical simulations. For example,
one can apply DMRG to find the ground state of $\mcH$, and then plug
it back to the Lindbladian to see if this is indeed the fixed point.
Since $\mcH$ can be easily obtained from $\mcL$, this procedure can
be used even if we do not know if the Lindbladian satisfies QDB.

It would also be interesting to further study the various necessary
conditions that are needed to prove an area-law for steady states of
local Lindbladians. In particular, it would be interesting find a
concrete example of a gapped detailed-balance Lindbladian that does not
obey rapid-mixing, thus separating our results from \cRef{Brandao}.
It would also be interesting to see if the conditions under which
our mapping applies can be relaxed. Is there a weaker condition that
still ensure a local $\mcH$?

Finally, it would also be interesting to understand if our mapping
can be used in the opposite direction. Given a local Hamiltonian,
one might ask if it is the vectorization of a super Hamiltonian that
comes from some QDB Lindbladian. In such cases it might be possible
to probe the ground state of the local Hamiltonian by simulating the
time evolution of the Lindbladian on a quantum computer. This might
show that the local Hamiltonian problem for this class of Hamiltonians
in inside BQP. It would be interesting to characterize this class of 
Hamiltonians, and see if they can lead to interesting quantum 
algorithms.

%%%%%%%%%%%%%%%%%%%%%%%%%%%%%%%%%%%%%%%%%%%%%%%%%%%%%%%%%%%%%%%%%%%%%
\section{Acknowledgements} 

We are thankful for Curt von Keyserlingk and Jens Eisert for
insightful discussions.
M.G.\ was supported by the Israel Science Foundation (ISF) and the
Directorate for Defense Research and Development (DDR\&D) Grant
No.~3427/21, and by the US-Israel Binational Science Foundation
(BSF) Grant No.~2020072. I.A.\ acknowledges the support of the Israel
Science Foundation (ISF) under the Individual Research Grant
No.~1778/17 and joint Israel-Singapore NRF-ISF Research Grant
No.~3528/20.

\bibliographystyle{ieeetr}
\bibliography{biblist.bib}

\begin{thebibliography}{10}

\bibitem{ref:Kitaev2002-QCbook}
A.~Y. Kitaev, A.~Shen, M.~N. Vyalyi, and M.~N. Vyalyi, {\em Classical and
  quantum computation}.
\newblock No.~47 in Graduate Studies in Mathematics, American Mathematical
  Soc., 2002.

\bibitem{ref:Kempe2006-QMA}
J.~Kempe, A.~Kitaev, and O.~Regev, ``{The Complexity of the Local Hamiltonian
  Problem},'' {\em SIAM Journal on Computing}, vol.~35, no.~5, pp.~1070--1097,
  2006.

\bibitem{ref:Eisert2010-AL-review}
J.~Eisert, M.~Cramer, and M.~B. Plenio, ``{Area laws for the entanglement
  entropy - a review},'' {\em Reviews of Modern Physics}, vol.~82,
  pp.~277--306, 2010.

\bibitem{ref:Hastings2007-1DAL}
M.~B. Hastings, ``An area law for one-dimensional quantum systems,'' {\em
  Journal of Statistical Mechanics: Theory and Experiment}, pp.~8024--8024, Aug
  2007.

\bibitem{ref:Masanes2009-AL}
L.~Masanes, ``{Area law for the entropy of low-energy states},'' {\em Phys.
  Rev. A}, vol.~80, p.~052104, Nov 2009.

\bibitem{ref:deBeaudrap2010-ALFFspins}
N.~de~Beaudrap, T.~J. Osborne, and J.~Eisert, ``{Ground states of unfrustrated
  spin Hamiltonians satisfy an area law},'' {\em New Journal of Physics},
  vol.~12, p.~095007, sep 2010.

\bibitem{ref:Spyridon2012-adiabaticAL}
S.~Michalakis, ``{Stability of the area law for the entropy of entanglement},''
  {\em arXiv preprint arXiv:1206.6900}, 2012.

\bibitem{ref:Jaeyoon2014-AL}
J.~Cho, ``{Sufficient Condition for Entanglement Area Laws in Thermodynamically
  Gapped Spin Systems},'' {\em Phys. Rev. Lett.}, vol.~113, p.~197204, Nov
  2014.

\bibitem{ref:Brandao2015-sp-heatAL}
F.~G. S.~L. Brand\~ao and M.~Cramer, ``{Entanglement area law from specific
  heat capacity},'' {\em Phys. Rev. B}, vol.~92, p.~115134, Sep 2015.

\bibitem{ref:Anshu2022-subvol}
A.~Anshu, I.~Arad, and D.~Gosset, ``{Entanglement Subvolume Law for 2D
  Frustration-Free Spin Systems},'' {\em Communications in Mathematical
  Physics}, vol.~393, Jul 2022.

\bibitem{ref:Anshu2021-2DAL}
A.~Anshu, I.~Arad, and D.~Gosset, ``{An area law for 2D frustration-free spin
  systems},'' {\em arXiv preprint arXiv:2103.02492}, 2021.

\bibitem{ref:Agarwal1973-QDB}
G.~Agarwal, ``Open quantum markovian systems and the microreversibility,'' {\em
  Zeitschrift f{\"u}r Physik A Hadrons and nuclei}, vol.~258, no.~5,
  pp.~409--422, 1973.

\bibitem{ALICKI1976249}
R.~Alicki, ``On the detailed balance condition for non-hamiltonian systems,''
  {\em Reports on Mathematical Physics}, vol.~10, no.~2, pp.~249--258, 1976.

\bibitem{Fangola}
F.~Fagnola and V.~Umanit{\`a}, ``Generators of detailed balance quantum markov
  semigroups,'' {\em Infinite Dimensional Analysis, Quantum Probability and
  Related Topics}, vol.~10, no.~03, pp.~335--363, 2007.

\bibitem{chi}
K.~Temme, M.~J. Kastoryano, M.~B. Ruskai, M.~M. Wolf, and F.~Verstraete, ``The
  $\chi^2$-divergence and mixing times of quantum markov processes,'' {\em
  Journal of Mathematical Physics}, vol.~51, no.~12, p.~122201, 2010.

\bibitem{Maas}
E.~A. Carlen and J.~Maas, ``Gradient flow and entropy inequalities for quantum
  markov semigroups with detailed balance,'' {\em Journal of Functional
  Analysis}, vol.~273, no.~5, pp.~1810--1869, 2017.

\bibitem{alicki2009}
R.~Alicki, M.~Fannes, and M.~Horodecki, ``On thermalization in kitaev's 2d
  model,'' {\em Journal of Physics A: Mathematical and Theoretical}, vol.~42,
  no.~6, p.~065303, 2009.

\bibitem{DOC}
M.~J. Kastoryano and J.~Eisert, ``Rapid mixing implies exponential decay of
  correlations,'' {\em Journal of Mathematical Physics}, vol.~54, no.~10,
  p.~102201, 2013.

\bibitem{kastoryano2016quantum}
M.~J. Kastoryano and F.~G. Brandao, ``Quantum gibbs samplers: The commuting
  case,'' {\em Communications in Mathematical Physics}, vol.~344, no.~3,
  pp.~915--957, 2016.

\bibitem{Szegedy}
P.~Wocjan and K.~Temme, ``Szegedy walk unitaries for quantum maps,'' {\em arXiv
  preprint arXiv:2107.07365}, 2021.

\bibitem{Brandao}
F.~G. S.~L. Brandão, T.~S. Cubitt, A.~Lucia, S.~Michalakis, and
  D.~Perez-Garcia, ``Area law for fixed points of rapidly mixing dissipative
  quantum systems,'' {\em Journal of Mathematical Physics}, vol.~56, no.~10,
  p.~102202, 2015.

\bibitem{simulatability}
R.~Trivedi and J.~I. Cirac, ``Simulatability of locally-interacting open
  quantum spin systems,'' {\em arXiv preprint arXiv:2110.10638}, 2021.

\bibitem{Swingle2016entanglement}
R.~Mahajan, C.~D. Freeman, S.~Mumford, N.~Tubman, and B.~Swingle,
  ``Entanglement structure of non-equilibrium steady states,'' {\em arXiv
  preprint arXiv:1608.05074}, 2016.

\bibitem{logsobolev}
M.~J. Kastoryano and K.~Temme, ``Quantum logarithmic sobolev inequalities and
  rapid mixing,'' {\em Journal of Mathematical Physics}, vol.~54, no.~5,
  p.~052202, 2013.

\bibitem{ref:Poulin2010-openLR}
D.~Poulin, ``{Lieb-Robinson Bound and Locality for General Markovian Quantum
  Dynamics},'' {\em Phys. Rev. Lett.}, vol.~104, p.~190401, May 2010.

\bibitem{ref:Nachtergaele2011-openLR}
B.~Nachtergaele, A.~Vershynina, and V.~A. Zagrebnov, {\em {Lieb-Robinson bounds
  and existence of the thermodynamic limit for a class of irreversible quantum
  dynamics.}}, vol.~552 of {\em Entropy and the quantum II.}
\newblock American Mathematical Soc., 2011.

\bibitem{ref:Barthel2012-openLR}
T.~Barthel and M.~Kliesch, ``{Quasilocality and Efficient Simulation of
  Markovian Quantum Dynamics},'' {\em Phys. Rev. Lett.}, vol.~108, p.~230504,
  Jun 2012.

\bibitem{ref:Arad2012-1DFFAL}
I.~Arad, Z.~Landau, and U.~Vazirani, ``{Improved one-dimensional area law for
  frustration-free systems},'' {\em Phys. Rev. B}, vol.~85, p.~195145, May
  2012.

\bibitem{ref:Arad2013-1DAL}
I.~Arad, A.~Kitaev, Z.~Landau, and U.~Vazirani, ``{An area law and
  sub-exponential algorithm for 1D systems},'' {\em arXiv preprint
  arXiv:1301.1162}, 2013.

\bibitem{ref:Kuwajara2020-long-range-AL}
T.~Kuwahara and K.~Saito, ``Area law of noncritical ground states in 1d
  long-range interacting systems,'' {\em Nature Communications}, vol.~11,
  no.~1, p.~4478, 2020.

\bibitem{bardet2021rapid}
I.~Bardet, {\'A}.~Capel, L.~Gao, A.~Lucia, D.~P{\'e}rez-Garc{\'\i}a, and
  C.~Rouz{\'e}, ``Rapid thermalization of spin chain commuting hamiltonians,''
  {\em arXiv preprint arXiv:2112.00593}, 2021.

\bibitem{bardet2021entropy}
I.~Bardet, {\'A}.~Capel, L.~Gao, A.~Lucia, D.~P{\'e}rez-Garc{\'\i}a, and
  C.~Rouz{\'e}, ``Entropy decay for davies semigroups of a one dimensional
  quantum lattice,'' {\em arXiv preprint arXiv:2112.00601}, 2021.

\bibitem{knabe}
S.~Knabe, ``{Energy gaps and elementary excitations for certain VBS-quantum
  antiferromagnets},'' {\em Journal of Statistical Physics}, vol.~52,
  pp.~627--638, Aug. 1988.

\bibitem{Lemm}
M.~Lemm, ``Gaplessness is not generic for translation-invariant spin chains,''
  {\em Phys. Rev. B}, vol.~100, p.~035113, Jul 2019.

\bibitem{ref:Breuer2006-opensys}
H.-P. Breuer and F.~Petruccione, ``Concepts and methods in the theory of open
  quantum systems,'' in {\em Irreversible Quantum Dynamics} (F.~Benatti and
  R.~Floreanini, eds.), pp.~65--79, Berlin, Heidelberg: Springer Berlin
  Heidelberg, 2003.

\bibitem{wolf}
M.~M. Wolf, ``Quantum channels \& operations: Guided tour,'' {\em Lecture notes
  available at
  \url{https://www-m5.ma.tum.de/foswiki/pub/M5/Allgemeines/MichaelWolf/QChannelLecture.pdf}},
  2012.

\bibitem{gaps_Znidari}
M.~\ifmmode \check{Z}\else \v{Z}\fi{}nidari\ifmmode~\check{c}\else \v{c}\fi{},
  ``Relaxation times of dissipative many-body quantum systems,'' {\em Phys.
  Rev. E}, vol.~92, p.~042143, Oct 2015.

\bibitem{cutoff}
M.~J. Kastoryano, D.~Reeb, and M.~M. Wolf, ``A cutoff phenomenon for quantum
  markov chains,'' {\em Journal of Physics A: Mathematical and Theoretical},
  vol.~45, p.~075307, feb 2012.

\bibitem{Sen_book}
E.~Seneta, {\em Non-negative Matrices and Markov Chains}.
\newblock Springer Series in Statistics, New York, NY: Springer New York, 2nd
  ed.~ed., 1981.

\bibitem{laracuente}
N.~LaRacuente, ``Self-restricting noise and quantum relative entropy decay,''
  {\em arXiv preprint arXiv:2203.03745}, 2022.

\bibitem{bolanos2013infinite}
J.~R. Bolanos-Servin and R.~Quezada, ``Infinite dimensional choi-jamiolkowski
  states and time reversed quantum markov semigroups,'' {\em arXiv preprint
  arXiv:1309.7091}, 2013.

\bibitem{szehr2015spectral}
O.~Szehr, D.~Reeb, and M.~M. Wolf, ``Spectral convergence bounds for classical
  and quantum markov processes,'' {\em Communications in Mathematical Physics},
  vol.~333, no.~2, pp.~565--595, 2015.

\bibitem{Gorini}
V.~Gorini, A.~Kossakowski, and E.~C.~G. Sudarshan, ``Completely positive
  dynamical semigroups of n‐level systems,'' {\em Journal of Mathematical
  Physics}, vol.~17, no.~5, pp.~821--825, 1976.

\bibitem{davies}
E.~B. Davies, ``Markovian master equations,'' {\em Communications in
  Mathematical Physics}, vol.~39, no.~2, pp.~91--110, 1974.

\bibitem{ref:Landau2015-1Dalg}
Z.~Landau, U.~Vazirani, and T.~Vidick, ``{A polynomial time algorithm for the
  ground state of one-dimensional gapped local Hamiltonians},'' {\em Nature
  Physics}, vol.~11, no.~7, pp.~566--569, 2015.

\bibitem{hastings2021gapped}
M.~B. Hastings, ``Gapped quantum systems: From higher dimensional
  lieb-schultz-mattis to the quantum hall effect,'' {\em arXiv preprint
  arXiv:2111.01854}, 2021.

\bibitem{ref:Watrous2018-QI}
J.~Watrous, {\em {The Theory of Quantum Information}}.
\newblock Cambridge University Press, 2018.

\bibitem{prosen2007operator}
T.~Prosen and I.~Pi{\v{z}}orn, ``Operator space entanglement entropy in a
  transverse ising chain,'' {\em Physical Review A}, vol.~76, no.~3, p.~032316,
  2007.

\bibitem{jonay2018coarse}
C.~Jonay, D.~A. Huse, and A.~Nahum, ``Coarse-grained dynamics of operator and
  state entanglement,'' {\em arXiv preprint arXiv:1803.00089}, 2018.

\bibitem{Has_LSM}
M.~B. Hastings, ``Lieb-schultz-mattis in higher dimensions,'' {\em Phys. Rev.
  B}, vol.~69, p.~104431, Mar 2004.

\bibitem{jauslin2021random}
I.~Jauslin and M.~Lemm, ``Random translation-invariant hamiltonians and their
  spectral gaps,'' {\em arXiv preprint arXiv:2111.06433}, 2021.

\bibitem{Moshe_fermions}
M.~Goldstein, ``{Dissipation-induced topological insulators: A no-go theorem
  and a recipe},'' {\em SciPost Phys.}, vol.~7, p.~67, 2019.

\bibitem{Shavit_Goldstein}
G.~Shavit and M.~Goldstein, ``Topology by dissipation: Transport properties,''
  {\em Phys. Rev. B}, vol.~101, p.~125412, Mar 2020.

\bibitem{Beck_Goldstein}
A.~Beck and M.~Goldstein, ``Disorder in dissipation-induced topological states:
  Evidence for a different type of localization transition,'' {\em Phys. Rev.
  B}, vol.~103, p.~L241401, Jun 2021.

\bibitem{Prosen_2008}
T.~Prosen, ``Third quantization: a general method to solve master equations for
  quadratic open fermi systems,'' {\em New Journal of Physics}, vol.~10,
  p.~043026, apr 2008.

\bibitem{barthel2021solving}
T.~Barthel and Y.~Zhang, ``Solving quasi-free and quadratic lindblad master
  equations for open fermionic and bosonic systems,'' {\em arXiv preprint
  arXiv:2112.08344}, 2021.

\bibitem{sanz2010quantum}
M.~Sanz, D.~Perez-Garcia, M.~M. Wolf, and J.~I. Cirac, ``A quantum version of
  wielandt's inequality,'' {\em IEEE Transactions on Information Theory},
  vol.~56, no.~9, pp.~4668--4673, 2010.

\bibitem{Toeplitz}
M.~S. Andreas~Frommer, Claudia~Schimmel, ``{Non-Toeplitz decay bounds for
  inverses of Hermitian positive definite tridiagonal matrices},'' in {\em
  {ETNA - Electronic Transactions on Numerical Analysis}} (L.~R.~H.
  Ronny~Ramlau, ed.), (Wien), pp.~362--372, Verlag der Österreichischen
  Akademie der Wissenschaften, 2018.

\bibitem{Crooks}
G.~E. Crooks, ``Quantum operation time reversal,'' {\em Phys. Rev. A}, vol.~77,
  p.~034101, Mar 2008.

\bibitem{lemm2020finite}
M.~Lemm, ``Finite-size criteria for spectral gaps in d-dimensional quantum spin
  systems,'' {\em Analytic trends in mathematical physics}, vol.~741, p.~121,
  2020.

\bibitem{Papoulis_book}
A.~Papoulis, {\em Probability, random variables, and stochastic processes /
  Athanasios Papoulis.}
\newblock McGraw-Hill series in electrical engineering. Communications and
  information theory, Auckland: McGraw-Hill, 2nd ed.~ed., 1984.

\end{thebibliography}

\appendix

%%%%%%%%%%%%%%%%%% App A %%%%%%%%%%%%%%%%%%%%%%%%%%
\section{Proof of \Lem{lem:exp-decay-matrix} }
\label{App:A}

In this appendix, we prove \Lem{lem:exp-decay-matrix}. For
convenience, we first restate it here.
\begin{lemma}
  Let $\lambda_{\mathrm{min}}>0$ be the smallest non-zero eigenvalue
  of $C\cdot C^*$, and assume that $|C_{ab}|\le J$ for every $a,b$. 
  Let $|a-b|$ denote the lattice distance between the supports of
  $P_a, P_b$. Then
  \begin{align*}
      |(C\cdot C^*)^{1/2}_{ab}| 
        \le c_1 J e^{-c_2|a-b|\cdot\lambda_{\mathrm{min}}/J^2}, 
  \end{align*}
  where $c_1, c_2$ are constants that depend only on the geometry of
  the lattice and on $k$.
\end{lemma}

\begin{proof}
  Set $A\EqDef C\cdot C^*$, and let $\lambda_{\mathrm{max}},
  \lambda_{\mathrm{min}}$ be the largest and smallest
  \emph{non-zero} eigenvalues of $A$. Let $|a-b|$ denote the lattice
  distance between the support of $P_a$ and $P_b$. Recall that
  $C_{ab}\neq 0$ only for $P_a, P_b$ that intersect a geometrically
  local region of $k$ sites, and therefore $C_{ab}\neq 0$ only for
  $|a-b|\le k$. Similarly, for any integer $\ell>0$,
  $C^\ell_{ab}\neq 0$ only when $|a-b|\le k\ell$, and therefore, as
  also $C^*_{ab}\neq 0$ only for $|a-b|\le k$, we conclude that
  $A^\ell_{ab} = (C\cdot C^*)^\ell_{ab}\neq 0$ only when $|a-b|\le
  2\ell k$.
  
  Following \cRef{Toeplitz}, we assume that there exists a family of
  polynomial approximations $\{P_m(z)\}$ to the function $\sqrt{z}$
  (indexed by their degree) with the following properties:
  \begin{enumerate}
    
    \item $|P_m(z)-\sqrt{z}|\le \tau_1 e^{-\tau_2 m}$ for every $z\in
      [\lambda_{\mathrm{min}}, \lambda_{\mathrm{max}}]$ for some
      constants $\tau_1,\tau_2$ that depend on
      $\lambda_{\mathrm{min}}, \lambda_{\mathrm{max}}$, but not on
      $m$.
      
    \item $P_m(0) = 0$.
  \end{enumerate}
  We will soon find such family, but for now let us discuss its consequences. 
  
  We first note that as the spectrum of $A$ is in $\{0\}\cup
  [\lambda_{\mathrm{min}}, \lambda_{\mathrm{max}}]$, then for any
  $m$, $\norm{\sqrt{A}-P_m(A)}\le \tau_1 e^{-\tau_2 m}$. Consider
  now a pair of indices $(a,b)$, and let $m$ be the largest integer
  for which $2k m < |a-b|$. Then by the discussion above
  $\big(P_m(A))_{ab} = 0$, and therefore by the triangle inequality,
  \begin{align*}
    |(\sqrt{A})_{ab}| &\le |(\sqrt{A})_{ab} - \big(P_m(A)\big)_{ab}| +
    |\big(P_m(A)\big)_{ab}| = |\big(\sqrt{A} - P_m(A)\big)_{ab}| \\
    &\le \norm{\sqrt{A} - P_m(A)} \le \tau_1e^{-m\tau_2}.
  \end{align*}
  However, as $m$ is the largest integer for which $2k m < |a-b|$, then
  $2k(m+1)\ge |a-b|$ and therefore $m\ge |a-b|/(2k) -1$, from which
  we deduce
  \begin{align} \label{eq:SAab1}
    |(\sqrt{A})_{ab}| \le \tau_1e^{\tau_2}
      \cdot e^{-|a-b|\cdot\tau_2/(2k)} .
  \end{align}
  
  Our next step is to show that the family of polynomials exists and
  find the dependence of $\tau_1, \tau_2$ on
  $\lambda_{\mathrm{min}}, \lambda_{\mathrm{max}}$. We will find an
  $m-1$ degree polynomial approximation to $1/\sqrt{z}$ in
  $[\lambda_{\mathrm{min}},\lambda_{\mathrm{max}}]$, and then
  multiply it by $z$. Following \cRef{ref:Anshu2022-subvol}, we use
  the expansion $1/\sqrt{1+z}=\sum_{j=0}^\infty \left( -\frac 1 4
  \right)^j \binom{2j}{j} z^j$ in the following manner:
  \begin{align}
  \label{eq:sqrt-z-expan}
    \frac{1}{\sqrt z} = \frac{1}{\sqrt{\lambda_{\mathrm{max}}}}\cdot
      \frac{1}{\sqrt{1+(z/\lambda_{\mathrm{max}}-1)}} 
      = \frac{1}{\sqrt{\lambda_{\mathrm{max}}}}\sum_{j=0}^\infty
    \left(-\frac{1}{4}\right)^j \cdot \binom{2j}{j}
      \cdot \left(\frac{z}{\lambda_{\mathrm{max}}}-1\right)^j .
  \end{align}
  For $z\in [\lambda_{\mathrm{min}},\lambda_{\mathrm{max}}]$, 
  $|z/\lambda_{\mathrm{max}}-1|<1$, and so the above series
  converges absolutely. Define $Q_m(z)$ to be the $m-1$ degree
  polynomial that is the sum of the terms in \eqref{eq:sqrt-z-expan}
  with degree $\le m-1$, and let $R_m(z)$ be the sum of all the
  higher order terms. Then $\frac{1}{\sqrt{z}} = Q_m(z) + R_m(z)$,
  and using the fact that $\binom{2j}{j}\leq 4^j$, we get
  \begin{align*}
    |R_m(z)|\leq \frac{1}{\sqrt{\lambda_{\mathrm{max}}}}
      \sum_{j=m}^\infty |z/\lambda_{\mathrm{max}}-1|^j
     = \frac{1}{\sqrt{\lambda_{\mathrm{max}}}} 
       \cdot \frac{|z/\lambda_{\mathrm{max}}-1|^m }{1-|z/\lambda_{\mathrm{max}}-1|}
    =\frac{\sqrt{\lambda_{\mathrm{max}}}}{z}
      \big(1-z/\lambda_{\mathrm{max}}\big)^m .
  \end{align*}
  Multiplying by $z$, and setting $P_m(z)\EqDef zQ_m(z)$, we find that
  \begin{align*}
    |\sqrt{z}-P_m(z)| \le \sqrt{\lambda_{\mathrm{max}}}
      \big(1-z/\lambda_{\mathrm{max}}\big)^m 
    \le \sqrt{\lambda_{\mathrm{max}}} e^{-mz/\lambda_{\mathrm{max}}}
    \le \sqrt{\lambda_{\mathrm{max}}} e^{-m\lambda_{\mathrm{min}}/\lambda_{\mathrm{max}}}.
  \end{align*}
  Therefore, $\tau_1=\sqrt{\lambda_{\mathrm{max}}},
  \tau_2=\lambda_{\mathrm{min}}/\lambda_{\mathrm{max}}\le 1$. Plugging to \eqref{eq:SAab1} yields
  \begin{align*}
    |(C\cdot C^*)^{1/2}_{ab}| \le e\cdot \sqrt{\lambda_{\mathrm{max}}} 
      \cdot e^{-|a-b|\lambda_{\mathrm{min}}/(2k\lambda_{\mathrm{max}})}.
  \end{align*}
  
  To conclude the proof, we need to show that
  $\lambda_{\mathrm{max}}\le \eta J^2$, where $\eta$ is a constant
  that is a function of the lattice geometry and $k$, independent of
  the system size. To do that, note that $\lambda_{\mathrm{max}} =
  \norm{C\cdot C^*} \le \norm{C}^2$. By definition, $C$ is a sparse
  matrix, since at every row $P_a$ there is only a constant number
  of $P_b$ that overlap the same geometrically $k$-local region.
  Call this constant $\sqrt{\eta}$, and note that it only depends on
  $k$ and the geometry of the lattice (i.e., dimension, etc.).
  Therefore, assuming that $|C_{ab}|\le J$ for all $a,b$, we deduce
  that for any normalized vector $v$, $\norm{Cv}\le J\sqrt{\eta}$,
  and so $\lambda_{\mathrm{max}}=\norm{C}^2\le \eta J^2$.
  
\end{proof}

% Let us restate it here for completeness.
% \begin{claim}\label{clm:banded}
% Let AA be a positive semi-definite kk-banded matrix with a non-zero spectrum in [a,b][a,b], then
% \begin{align*}
%     |\sqrt{A}_{i j}|\leq K e^{-c|i-j|}
% \end{align*}
% where c=ln(√κ−1√κ+1)c =\ln\left(\frac{\sqrt{\kappa}-1}{\sqrt{\kappa} +1}\right) and κ=b/a\kappa=b/a.
% \end{claim}

%%%%%%%%%%%%%%%%%% App C %%%%%%%%%%%%%%%%%%%%%%%%%%
\section{Classical-like Lindbladian: Technical Details}
\label{App:toy1}

This appendix is devoted to proving bullets~\ref{clm:bul3},
\ref{clm:bul4} in Claim~\ref{clm:local_toy_model}.

%====================================================================
\subsection{Proving uniqueness of the steady state
(bullet \ref{clm:bul3})}
\label{sec:uniqueness_model1}

We prove uniqueness of the steady state in two steps. First, we show that the steady
state of the restriction of $\mcL$ to the diagonal elements is
unique. Then we show that all off-diagonal elements decay due to
the dissipative dynamics.  We start by writing the action of the
Lindbladian on a general diagonal element $\ketbra x x$ where
$x\in\{0,1\}^n$:
\begin{align*}
\mcL (\ketbra x x) = \sum_{k,\bi} \gamma_{k,b} \LLkl  (\ketbra x x)
= \sum_{x^\prime\sim x} f_{x\rightarrow x^\prime}\ketbra{x^\prime}{x^\prime} -g_x \ketbra x x,
\end{align*}
where the sum runs over all strings that can be achieved from $x$ by flipping one spin.
The corresponding weights are give by:
\begin{align*}
f_{x\rightarrow x^\prime} & =\gamma_{k,\bi_k}\cdot
\begin{cases}
e^{-\beta\omega_{k,b}/2}, & x_k=1,\\
e^{\beta\omega_{k,b}/2}, & x_k = 0,
\end{cases}\\
g_x &= \sum_{x^\prime \: ; \: x\rightarrow x^\prime}f_{x\rightarrow x^\prime},
 %+\sum_{k \: ; \: n_k=0}\gamma_{k,b} e^{\omega_{k,\bi_k}},
\end{align*} 
where $\bi_k=\bi_k(x)=x_{k-1}+x_{k+1}$, and $k$ is defined by the spin that is flipped when $x\rightarrow x^\prime$.
Note that $g_x$ is responsible for ensuring that
$\Tr\Big(\mcL(\ketbra x x )\Big)=0$.

It is known that a Lindbladian, being a generator of a CPTP
semigroup, induces a (classical) continuous time Markov Process on
the diagonal elements $\{\ketbra{x}{x} \}$ which is defined
by\cc{Crooks}
\begin{align*}
  P_{xy}(t)\EqDef \Tr(\ketbra{x}{x} \cdot e^{\mcL t}(\ketbra{y}{y})) .
\end{align*}
To see this, Notice that 
$\Tr\big[\ketbra x x \cdot e^{\mcL t}(\ketbra y y)\big] \geq 0 $
due to $e^{\mcL t}$ being completely positive.
One can also check the sum of each column to see that
\begin{align*}
    \sum_{x}P_{xy}(t) = \sum_x \Tr\big[\ketbra x x 
      \cdot e^{\mcL t}(\ketbra y y)\big] = \Tr\big[e^{\mcL t}(\ketbra y y)\big] = 1,
\end{align*}
due to $e^{\mcL t}$ being trace preserving.
To deduce the uniqueness of a stationary distribution of $P_{xy}$,
we use the Perron-Frobenious Theorem and the connectivity of the Markov chain.
Specifically, we show that $P_{x y}(t)>0$ for any $x,y,t$, and uniqueness will follow as a consequence from Perron-Frobenious (see Theorem 1.1 in \cRef{Sen_book}).
\begin{claim}
  $P_{x y}(t)=\Tr\big[\ketbra x x e^{\mcL t}(\ketbra y y)\big]>0$ for any $t>0$.
\end{claim}
\begin{proof}
Let $G=(V,E)$ be a graph where $V=\{0,1\}^n$ is the set of bit-strings and $(x,x^\prime)\in E$ if and only if they have Hamming distance $d_H(x,x^\prime)= 1$ ($x^\prime$ is obtained from $x$ by flipping one bit). Notice that this is also the connectivity graph of $\mcL$ with self edges omitted. 
For $x,y\in V$, define $\ell=(\ell_0,\ldots,\ell_{|\ell|+1})$ to be the shortest path in $G$ from $x$ to $y$ (i.e. $d_H(x,y)=|\ell |$).
Take $c>0$ such that $c>g_z$ for any $z\in V$. As a result,
$(\mcL t+c\Id)_{x y}\geq 0$ for any $x,y$, and in particular $(\mcL+c\Id)_{xy}>0$ if $(x,y)\in E$.
Let us expand  $e^{\mcL t}$
\begin{align}
e^{\mcL t} = e^{-c}e^{\mcL t+c\Id} = e^{-c}\sum_{k=0}^\infty \frac{1}{k!}(\mcL t+c\Id)^k .
\end{align}
Notice that $\left( (\mcL t+c\Id)^k\right)_{x y}\neq 0$ if and only if $d_H(x,y)\leq k$, therefore $P_{x y}$ receives its first non-zero contributions from the 
$k=|\ell |$th order in the expansion.
Moreover, this first contribution is positive, since
\begin{align*}
\braket{x |(\mcL t+c\Id)^k | y } & 
= \sum_{r_1,\ldots ,r_{k-1}}\braket{ x | \mcL t+c\Id |r_1 }
\braket{r_1 | \mcL t+c\Id |r_2 }\ldots \braket{r_{k-1}| \mcL t+c\Id |y}\\
& =   \braket{ x | \mcL t+c\Id |\ell_1 }
\braket{\ell_1 | \mcL t+c\Id |\ell_2 }\ldots \braket{\ell_{k-1}| \mcL t+c\Id |y}\\
&
+ \sum_{\{r_1,\ldots ,r_{k-1}\}\neq \ell}\braket{ x | \mcL t+c\Id |r_1 }
\braket{r_1 | \mcL t+c\Id |r_2 }\ldots \braket{r_{k-1}| \mcL t+c\Id |y}>0,
\end{align*}
where we abuse of notation by writing $\braket{x|\mathcal{T}|y}$
instead of $\Tr\big[\ketbra x x \cdot \mathcal{T}(\ketbra y y)\big]$,
exploiting the fact that the action of $\mathcal T=\mcL t+c\Id$ is
closed on the diagonal elements $\ketbra x x $.  Notice that the
first term after the equality sign is greater than zero using the path
$\ell$ and graph connectivity of the semigroup, and the rest of the
sum is greater than or equal to zero, hence the whole term is
greater than zero. The higher $k$ terms are greater or equal to zero
(since all matrix elements are), thus we conclude the claim.
\end{proof}
The only thing left to show now is that the off-diagonal elements decay.
Let us write the action of the Lindbladian on the off diagonal terms explicitly:
\begin{align*}
\mcL (\ketbra x y) = \sum_{x^\prime, y^\prime} f_{x y\rightarrow x^\prime y^\prime}
\ketbra{x^\prime}{y^\prime}- \left(\frac{1}{2}g_x + \frac{1}{2}g_y\right) \ketbra x y,
\end{align*}
where $f_{x y\rightarrow x^\prime y^\prime}$ is non-zero only if there are $k\in \{1,\ldots,n\}$ and $\bi\in\{0,1,2\}$  such that $L_{k,\bi}\ketbra x y L_{k,b} ^\dagger \neq 0$ or $L_{k,\bi}^\dagger\ketbra x y L_{k,b}  \neq 0$
(that is, if $x_k=y_k$, $x_{k-1}+x_{k+1}=y_{k-1}+y_{k+1}$), and then it would give $f_{x y\rightarrow x^\prime y^\prime}=f_{x\rightarrow x^\prime}=f_{y\rightarrow y^\prime}$.
The observation is that 
\begin{align*}
    \sum_{(x^\prime, y^\prime)} f_{x y\rightarrow x^\prime y^\prime}<\frac{1}{2}g_x+ \frac{1}{2}g_y,
\end{align*}
namely, looking at the representing matrix $\mcL |_{\{\ketbra m n\}_{m\neq n}}$, the sum of the off-diagonal elements in each column is strictly smaller than the absolute value of the element on the diagonal.
We claim that the eigenvalues of such matrix must be strictly negative.
\begin{claim}
Let $A_{ij}$ be a matrix with non-negative off diagonal elements and strictly negative diagonal elements, such that $\sum_{j \; ; \; j\neq i} A_{ij}<-A_{ii}$ for any $i$. Then $\mathrm{Spec}(A)\subset (-\infty,0)$.
\end{claim}
% Note that this claim applies on (\mcL|{\ketbraxy}x≠y)t\left(\mcL |_{\{\ketbra x y\}_{x\neq y}}\right)^t, which obviously implies the same for \mcL|{\ketbraxy}x≠y\mcL |_{\{\ketbra x y\}_{x\neq y}}.
\begin{proof}
Consider the exponentiation of $A$. This is a semi-stochastic matrix (its matrix-elements are non-negative), since
\begin{align*}
e^{A} = e^{-c}e^{A+c\Id} =e^{-c}\sum_{k=0}^\infty \frac{1}{k!}\left(A+c\Id\right)^k
\end{align*}
where $c>0$ is chosen such that $A_{ii}+c>0$ for each $i$.
Therefore, $(A+c\Id)_{i,j}$, and correspondingly $(e^A)_{ij}$, are greater or equal to zero for any $i,j$.
Using the Perron-Frobenious Theorem, $e^{A}$ has a maximal eigenvalue $r>0$ with an eigenvector $v$ in which $v_i\geq 0$ for any $i$.
We deduce that $v$ is also an eigenvector of $A$ with largest eigenvalue, due to the monotonicity of the exponent:
\begin{align}
\label{eq:matrix_eq}
\sum_{j}A_{ij}v_j=\lambda v_i \tab \forall i.
\end{align}
We finally show that $\lambda$ is negative due to the assumption in the claim, and the remaining spectrum will be negative as well (since it  must be below $\lambda$).
Indeed, let $i_0=\arg \max\{v_i\}$, and take
\begin{align*}
\lambda v_{i_0}=\sum_j A_{i_0j}v_j=A_{i_0 i_0} v_{i_0} +\sum_{j\neq i_0} A_{i_0 j}v_j \leq A_{i_0 i_0} v_{i_0} + v_{i_0}\sum_{j\neq i_0} A_{i_0 j}< 0,
\end{align*}
where in the last inequality we used $\sum_{j\neq i} A_{ij}<-A_{ii}$.
\end{proof}

%====================================================================
\subsection{Proving that $\mathcal{L}$ is gapped (bullet  \ref{clm:bul4})}
\label{sec:gap_model1}

To show a gap in the system, we import the \emph{finite-size
criteria} for frustration-free Hamiltonians, originally introduced
by Knabe\cc{knabe} and used, for example, in 
Refs.~\cc{Lemm,jauslin2021random}. First we
introduce the method for generic systems defined on finite
dimensional Hilbert spaces, and then use it to show a gap for the
system under consideration. We change our notation correspondingly,
e.g., we first discuss projectors on a generic Hilbert space and
denote them by $P_e$, and in the next paragraph we refer to
projectors on operators space (super-projectors) which we denote by
$\mcP_e$.

%%% the method %%%
\begin{claim}[Finite size criteria]
\label{clm:finite_size_cri}
Let $G=(V,E)$ be a regular graph of degree $\delta$. Assign a local Hilbert-space of dimension $d$ to each vertex. Let $H$
be a nearest-neighbors frustration free Hamiltonian defined on the joint Hilbert space
of the vertices of the form:
\begin{align} 
\label{eq:sumofproj}
H=\sum_{e\in E} P_{e},
\end{align}
where each $P_e$ is a projector defined on the bond $e$.
We define the local gap to be $\gamma_{loc}\EqDef \min_{e\cap e' \neq \emptyset} \left( \gap\left[ P_e+P_{e'}\right]\right)$.
Then
\begin{align} \label{eq:Knabe's_gap}
    \gap{(H)}\geq 2(\delta-1) \left(\gamma_{loc}-\left(1-1/2(\delta-1) \right)\right).
\end{align}
\end{claim}
\begin{proof}
Note that due to frustration-freeness and spectral decomposition,
an inequality of the form 
\begin{align} 
    \label{eq:inequality_gap_c}
    H^2\geq c\cdot H
\end{align} 
immediately implies that $\gap(H)\geq c$.
This is the principle that lies at heart of the method.
To achieve a bound of the form (\ref{eq:inequality_gap_c}), Knabe used the Hamiltonian's local structure and the local spectral gap, as explained in the following.
We start by squaring $H$ and rearranging the double sum:
\begin{align}
\label{eq:squaringH}
H^2= \left( \sum_e P_e\right)^2 = \sum_{e} P_e ^2
+ \sum_{e\cap e'\neq \emptyset} (P_e\cdot P_{e'} + P_{e'}\cdot P_{e})
+ 2\sum_{e\cap e'=\emptyset} P_e\cdot P_{e'}.
\end{align}
Since we are looking for an inequality of the form $H^2\geq c H$, we can ignore the rightmost term in (\ref{eq:squaringH}), as it is positive semi-definite, being a sum of products of non-overlapping projectors.
Using $P_e^2=P_e$, one should notice that the first term in \Eq{eq:squaringH} reduces to $H$ itself.
Hence we obtained 
$H^2\geq H+Q $, where  $Q\EqDef \sum_{e\cap e'\neq \emptyset}P_e\cdot P_{e'}+P_{e'}\cdot P_{e}$.
To achieve the desired bound, we write down a lower bound for $Q$ using $H$.
To do so, let us consider
\begin{align*}
\sum_{e\cap e'\neq \emptyset}\left(P_e+ P_{e'}\right)^2.    
\end{align*}
Expanding the square, we see that the expression equals 
\begin{align*}
    \sum_{e\cap e'\neq \emptyset}\left(P_e+ P_{e'}+P_e P_{e'}+P_{e'} P_{e}\right) = 2(\delta-1) \cdot H+Q,
\end{align*}
where the $2(\delta-1)$ factor stems from the fact that for a simple regular graph, each edge intersects $2(\delta-1)$ distinct edges ($\delta-1$ on each vertex).
On the other hand, due to frustration freeness, we have $\left(P_e+ P_{e'}\right)^2\geq \gap\left[P_e+ P_{e'}\right]\left(P_e+ P_{e'}\right)$, and obtain
\begin{align*}
2(\delta-1) \cdot H + Q \geq 2(\delta-1)\gamma_{loc}  \cdot H.
\end{align*}
Plugging this into \Eq{eq:squaringH} gives
\begin{align} \label{eq:finite_size_criteria}
H^2 \geq H+Q \geq 2(\delta-1) \cdot H\left(\gamma_{loc}-\left(1-\frac{1}{2(\delta-1 )}\right)\right).
\end{align}
Using inequality \eqref{eq:inequality_gap_c} ($H^2\geq c \cdot H$) for $c=2(\delta-1)\left(\gamma_{loc}-\left(1-\frac{1}{2(\delta-1)}\right)\right)$, we see that
\begin{align*}
    \gap(H) \geq 2(\delta-1)\left(\gamma_{loc}-\left(1-\frac{1}{2(\delta-1)}\right)\right),
\end{align*}
meaning that if $\gamma_{loc}$ is greater than $1-\frac{1}{2(\delta-1)}$, a global gap in $H$ confirmed.
\end{proof}
We remark that the method can be generalized to larger sub-regions and to open boundary conditions\cc{knabe,lemm2020finite}.

%%% the implementation %%%
\subsubsection*{Application to $\mcH$}
Now we apply the method to the system under consideration to show a global gap in $\mcL$, thus completing bullet~\ref{clm:bul4} of Claim~\ref{clm:local_toy_model}.
First, we reduce the super-Hamiltonian to the sum-of-local-projectors Hamiltonian presented in \Eq{eq:sumofproj} (up to a multiplicative factor). Then, we numerically verify the statement that $\gamma_{loc}>1-\frac{1}{2(\delta-1)}=\frac{1}{2}$ for the 1D ring, which will assure the global gap by \Eq{eq:finite_size_criteria}.

Recall the parent Hamiltonian, $\mcH= \sum_{k,\bi} \gamma_{k,b}\hkl$,
where $\hkl=-\Gamma_s^{-1/2}\circ \LLkl \circ \Gamma_s^{1/2}$ and $\gkl\geq \alpha >0$.
% Explicitly:
% \begin{align*}
% \hkl(A ) =  - L_{k,b} A L_{k,b}^\dagger - L_{k,b}^\dagger A L_{k,b}
% + \frac{e^{-\beta\omega_{k,b}/2}}{2}\{L_{k,b}^\dagger L_{k,b},A\} 
% + \frac{e^{\beta\omega_{k,b}/2}}{2}\left\{ L_{k,b} L_{k,b}^\dagger , A \right\}.
% \end{align*}
Also recall that $\mcH$ is frustration-free, where
$\hkl\geq 0$ has a ground state $\sqrt \sigma$ with energy $0$.
Then $\mcH$ satisfies
\begin{align*}
\mcH & \geq \alpha \sum_{k,\bi}\hkl
= \alpha \sum_{k}\mch_{k},\\
\mathrm{where}  &\tab \mch_{k}\EqDef \sum_{\bi=0,1,2} \hkl.
% = \sum_a \left[e^{\beta\omega_{k,a}/2} L_{k,a} A L_{k,a}^\dagger -\frac{1}{2}\{L_{k,a}^\dagger L_{k,a},A\}
%+ e^{\beta\omega_{k,a}/2}  L_k^\dagger A L_k -\frac{1}{2}e^{\beta\omega_{k,a}}\left\{ L_k L_k^\dagger , A \right\}\right].
\end{align*}
Let $\mcP_k$ denote the local projector onto the positive spectrum of $\mch_k$ such that $\mch_k\geq g_k \mcP_k$ where $g_k \EqDef \gap(\mch_k)$.
By simply calculating the spectrum of $\mch_k$ we see that $g_k =\min_{\bi\in\{0,1,2\}} \{ e^{\beta\omega_{k,b}} \}=e^{-\beta(\epsilon_k-\mu+2u)}$, which is of order $1$ provided that $\epsilon_k-\mu, u = \bigO{1}$.
We thus conclude that %this paragraph to the statement that
\begin{align*}
\mcH\geq \alpha g_{\min} \sum_k \mcP_k,
\end{align*}
where we defined $g_{\min} \EqDef\min_{k} g_k$.
Moreover, since both sides have a common ground space which is spannded by $\sqrt \sigma$ with ground energy $0$, using Claim~\ref{clm:finite_size_cri} we have
\begin{align*}
    \gap(\mcH) \geq \alpha g_{\min} \gap\left(\sum_k \mcP_k\right) \geq 2\alpha g_{\min} (\gamma_{loc}-1/2),
\end{align*}
where $\gamma_{loc}=\min_k\gap(\mcP_k+\mcP_{k+1})$.
To show a gap in $\mcH$, we assure that $\gap(\mcP_k+\mcP_{k+1})> 1/2$.
Instead of giving a tedious derivation of the local gap,
we numerically verify it in the following table for $\mu=0$ and different choices of\footnote{In the first line of the table, the random energies $\{ \epsilon_k \}$ were picked from the uniform distribution on $[0,1]$, and the local gap was averaged over 100 instances.} $\{\epsilon_k\}$ and $u$.

\begin{center}
\begin{tabular}{ |p{3cm}||p{2cm}|p{2cm}|p{2cm}|}
 \hline
 \multicolumn{4}{|c|}{Local Gap $\gamma_{loc}$}\\
\hline
Model & u=-1 & u=0.5 & u=2 \\
\hline
Random $\{ \epsilon_k \}$ & 0.756 & 0.887 & 0.678 \\
Const $\epsilon_k=1$ & 0.769 & 0.913 & 0.778 \\
Const $\epsilon_k=0.5$ &  0.755 & 0.891 & 0.698 \\
$\epsilon_1=1$, $\epsilon_2=10$ & 0.993 & 0.998 & 0.997 \\
\hline
\end{tabular}
\end{center}

\vspace{2mm}
We remark that the derivation of the finite size criteria considers 2-local interaction (graph interactions), and the model under consideration here is 3-local.
However, the derivation and results are still valid. This is due to terms that overlap on a single qubit (e.g. $\mcH_k$ and $\mcH_{k+2}$) commute with each other, as they act trivially on their shared qubits.
Therefore, when squaring $\mcH$ in \Eq{eq:squaringH}, their product can be absorbed in the positive "leftover" (the last term in \eqref{eq:squaringH}).

%%%%%%%%%%%%%%%%%% App C %%%%%%%%%%%%%%%%%%%%%%%%%%
\section{Quadratic fermionic Lindbladians:  Technical Details }
\label{App:toy2}

This appendix is devoted to proving  \Lem{lem:gap_in_gamma}, which demonstrates the exponential decay in
the coefficients of the super-Hamiltonian with respect to the distance
between the real space fermionic operators.

First, recall the settings of the lemma.
To relate the model  to the subject of our work, which discusses local Lindbladians that satisfy quantum detailed-balance,
we focus on commuting $\gin,\gout$ that connect nearest neighbours (i.e. tridiagonal matrices).
For simplicity, we consider translation invariant systems, such that
\begin{align*}
\gin =
\begin{pmatrix}
\gin_0 & \gin_1 & 0 & & \hdots & 0 & \gin_1 \\
\gin_1 & \gin_0 & \gin_1 & & & & 0 \\
0 & \gin_1  &\gin_0   & \\
\vdots & &  &  \ddots &  & & \vdots \\
& & &  & \gin_0 & \gin_1 & 0\\
0 & & & &\gin_1 & \gin_0 & \gin_1 \\
\gin_1 & 0 & & \hdots  & 0 & \gin_1 & \gin_0 
\end{pmatrix}, \quad
\gout = \begin{pmatrix}
\gout_0 & \gout_1 & 0 & & \hdots & 0 & \gout_1 \\
\gout_1 & \gout_0 & \gout_1 & & & & 0 \\
0 & \gout_1  &\gout_0   & \\
\vdots & &  &  \ddots &  & & \vdots \\
& & &  & \gout_0 & \gout_1 & 0\\
0 & & & &\gout_1 & \gout_0 & \gout_1 \\
\gout_1 & 0 & & \hdots  & 0 & \gout_1 & \gout_0 
\end{pmatrix}.
\end{align*}
The canonical jump operators are given by a discrete Fourier transform of the original $a_k$:
\begin{align*}
    c_k = \frac{1}{\sqrt{n}}\sum_j e^{-\frac{2\pi i}{n}k  j}a_j, \quad k=1\cdots n .
\end{align*} 
These are non-local operators, but a linear combination of such.

To prove the lemma, we want to show that the matrix $\sqrt{\gin}\cdot \sqrt{\gout}$ has elements that decay exponentially with the distance from the main diagonal.
First, we show that $\sqrt{ \gin}$ and $\sqrt{ \gout}$ individually satisfy
\begin{align*}
    (\sqrt {\gin})_{i j} = O\left(\frac{\sqrt{\gin_0}}{1-\frac{2\gin_1}{\gin_0}}\left(\frac{2\gin_1}{\gin_0}\right)^{d(i,j)}\right), \tab
    (\sqrt {\gout})_{i j} = 
    O\left(\frac{\sqrt{\gout_0}}{1-\frac{2\gout_1}{\gout_0}}
    \left(\frac{2\gout_1}{\gout_0}\right)^{d(i,j)}\right).
\end{align*}
Then, we use this to establish an upper-bound on $\sqrt{\gin\cdot\gout}_{i j}$.

To show a decay in $\gin$ (the treatment of $\gout$ is identical),
we notice that
\begin{align*}
\gin = \begin{pmatrix}
\gin_0 &  & & & \\
 & \gin_0 &   & & \\
& & \ddots & & \\
& & &  & \gin_0 &  \\
& & & &  & \gin_0 
\end{pmatrix}
+
 \begin{pmatrix}
0 & \gin_1 &  & & \gin_1\\
\gin_1 & 0& \gin_1  & \\
 & & \ddots & & \\
 & &  \gin_1 & 0 & \gin_1 \\
\gin_1 & & & \gin_1 & 0
\end{pmatrix} \EqDef \gin_0 \Id+ 2\gin_1 W .
\end{align*}
A key observation is that $W$ is the (symmetric) random walk matrix $W_{ij} = \frac{1}{2}\delta_{j,i\pm 1}$.
Plugging the Taylor expansion $\sqrt{1+x}=\sum_{k\geq 0} \alpha_k x^k$  produces
\begin{align*}
\sqrt{\gin}= \sqrt{\gin_0 \Id+ 2 \gin_1 W} =
\sqrt{\gin_0} \sqrt{\Id+ \frac{2
\gin_1}{\gin_0} W} = \sqrt{\gin_0} \sum_{k\geq 0} \alpha_k \left(\frac{2\gin_1}{\gin_0}\right)^k W^k .
\end{align*}

Note that $(W^k)_{ij}$ is the transition probability for $i\rightarrow j$ after $k$ steps of the walk, which is less than one and non-zero only when $k\geq \min\{|i-j|,n-|i-j|\}$;  Assume without loss of generality that the minimum is attained at $|i-j|$. Moreover, the sequence $\{\alpha_k\}$ is smaller than one in absolute value (polynomially decaying to be precise).
Then the matrix entry $(\sqrt{\gin})_{i,j}$ satisfies
\begin{align}
\label{eq:Bij}
(\sqrt{\gin})_{i,j} = \sqrt{\gin_0} \sum_{k\geq 0} \alpha_k \left(\frac{2\gin_1}{\gin_0}\right)^n (W^k)_{ij} \leq  \sqrt{\gin_0} \sum_{k\geq |i-j|} \left(\frac{2\gin_1}{\gin_0}\right)^k =  
\left(\frac{2\gin_1}{\gin_0}\right)^{|i-j|}
\left(\frac{\sqrt{\gin_0}}{1-\frac{2\gin_1}{\gin_0}} \right).
\end{align}
This shows that $(\sqrt{\gin})_{ij}=\bigO{e^{-|i-j|}}$ provided that $\gin$ is gapped and $\norm{\gin}=\bigO{1}$. The analogous argument  for $\gout$ is straightforward.
Note that for $\gin_0 = 2\gin_1$,  the gap in $\gin$ closes, and  a polynomial decay in $\sqrt{\gin}_{ij}$ appears instead (see \autoref{fig:graphs}).
This is because for large $k$'s,
the probability for the random walk yields $(W^k)_{ij}=
\frac{1}{2^{k+1}}{\binom{k}{\frac{|i-j|+k}{2}}}=\bigO{k^{-1/2}}$ (see page 291 in \cRef{Papoulis_book}), and also we get that\footnote{This can be derived using the expression $\alpha_n = \frac{(2n-3)!}{2^{2n-2} n! (n-2)!}$ and applying the Stirling's formula for large $n$.} $\alpha_k = \bigO{\frac{1}{k}}$.

\begin{figure}
    \centering
    \includegraphics{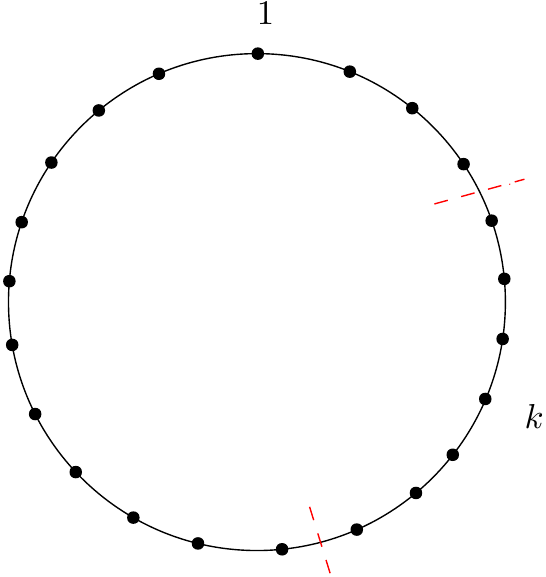}
    \caption{The neighborhoods for the partial sums in the
      calculation for $B_{1 k}$.}
    \label{fig:circle_sums}
\end{figure}

It is now left to show that the product $B\EqDef \sqrt{\gin\cdot\gout}$ has exponentially decaying matrix elements as well.
We show this for first row elements $B_{1 k}$ for convenience, as it automatically generalized to any $B_{i j}$ due to translation invariance. 
Let us write down $B_{1 k}$, explicitly and use the former bounds for $\sqrt{\gin}$ and $\sqrt{\gout}$. Suppose that $k\leq n/2$ as before; then
\begin{align*}
    |B_{1 k}| & \leq
    \sum_m |\sqrt{\gin}_{1 m} \sqrt{\gout}_{m k}|\\
    & \leq 
    \sum_m \frac{\sqrt{\gin_0}}{1-\frac{2\gin_1}{\gin_0}}
    \frac{\sqrt{\gout_0}}{1-\frac{2\gout_1}{\gout_0}}
    \left(\frac{2\gin_1}{\gin_0}\right)^{d(m,1)}
    \left(\frac{2\gout_1}{\gout_0}\right)^{d(k,m)}\\
    & = \frac{\sqrt{\gin_0}}{1-\frac{2\gin_1}{\gin_0}}
    \frac{\sqrt{\gout_0}}{1-\frac{2\gout_1}{\gout_0}} \left[
    \sum_{m=k/2+1}^{k+k/2}\left(\frac{2\gin_1}{\gin_0}\right)^{d(m,1)}
    \left(\frac{2\gout_1}{\gout_0}\right)^{d(k,m)} +
    \sum_{m=3 k/2+1}^{k/2}
    \left(\frac{2\gin_1}{\gin_0}\right)^{d(m,1)}
    \left(\frac{2\gout_1}{\gout_0}\right)^{d(k,m)}\right],
\end{align*}
where the last sum is performed periodically, i.e. $m =
3k/2+1,3k/2+2,\ldots , n-1, n, 1, 2,\ldots , k/2$.  All we did is to
separate the sum to a $k/2$ length neighbourhood of the $k$th site
and the complement (see \autoref{fig:circle_sums}). At the first
sum, the decay is dominated by
$\left(\frac{2\gin_1}{\gin_0}\right)^{d(m,1)}
=\left(\frac{2\gin_1}{\gin_0}\right)^{m-1} \leq
\left(\frac{2\gin_1}{\gin_0}\right)^{k/2}$, and at the second
sum the $\left(\frac{2\gout_1}{\gout_0}\right)^{d(k,m)}\leq
\left(\frac{2\gout_1}{\gout_0}\right)^{k/2}$ term dominates.
Therefore we write
\begin{align*}
     |B_{1 k}| \leq  \frac{\sqrt{\gin_0}}{1-\frac{2\gin_1}{\gin_0}}
    \frac{\sqrt{\gout_0}}{1-\frac{2\gout_1}{\gout_0}} \left[
    \sum_{m=k/2+1}^{k+k/2}\left(\frac{2\gin_1}{\gin_0}\right)^{m-1}
     +
    \sum_{m=3 k/2+1}^{k/2}
    \left(\frac{2\gout_1}{\gout_0}\right)^{d(k,m)}\right],
\end{align*}
where we used the fact that both $\frac{2\gin_1}{\gin_0}$ and $\frac{2\gout_1}{\gout_0}$ are smaller than $1$.
Since both sums are  geometric, they can be bounded to achieve the desired bound
\begin{align*}
     |B_{1 k}| \leq  2\frac{\sqrt{\gin_0}}{\left(1-\frac{2\gin_1}{\gin_0}\right)^2}
    \frac{\sqrt{\gout_0}}{\left(1-\frac{2\gout_1}{\gout_0}\right)^2} \left[
    \left(\frac{2\gin_1}{\gin_0}\right)^{k/2}
     +
     \left(\frac{2\gout_1}{\gout_0}\right)^{k/2}\right] .
\end{align*}

\end{document}